\documentclass[12pt, draftclsnofoot, onecolumn]{IEEEtran}
\ifCLASSINFOpdf

\else

\fi
\usepackage{amsmath}
\usepackage{makeidx}  
\usepackage{algorithm}
\usepackage{algorithmic}
\usepackage{graphicx}
\usepackage{subfigure}
\usepackage{epstopdf}
\usepackage{bm}
\usepackage{bbding}
\usepackage{cite}
\usepackage{stfloats}

\usepackage{amssymb}
\usepackage{graphicx}

\usepackage{url}
\newtheorem{proposition}{Proposition}
\usepackage{tabularx,booktabs}
\newcolumntype{C}{>{\centering\arraybackslash}X} 
\usepackage{lipsum}

\usepackage{makecell} 

\usepackage{graphicx}
\usepackage{color}
\newtheorem{thm}{Theorem}
\newtheorem{rem}{Remark}
\newtheorem{lem}{Lemma}
\newtheorem{proof}{proof}

\begin{document}
\linespread{1.48}

\title{Active IRS Aided Multiple Access for Energy-Constrained IoT Systems}

\author{Guangji~Chen,\IEEEmembership{}
        Qingqing~Wu,\IEEEmembership{}
        Chong~He,\IEEEmembership{}
        Wen~Chen,\IEEEmembership{}
        Jie~Tang,\IEEEmembership{}
        and Shi~Jin\IEEEmembership{}
        \thanks{G. Chen and Q. Wu are with the State Key Laboratory of Internet of Things for Smart City, University of Macau, Macao 999078, China (email: guangjichen@um.edu.mo; qingqingwu@um.edu.mo). C. He and W. Chen are with the Department of Electronic Engineering, Shanghai
        Jiao Tong University, Shanghai 200240, China (e-mail: hechong@sjtu.edu.cn; wenchen@sjtu.edu.cn).  J. Tang is with the School of
        Electronic and Information Engineering, South China University of Technology, Guangzhou 510641, China (email: eejtang@scut.edu.cn). S. Jin is with the National Mobile Communications
        Research Laboratory, Southeast University, Nanjing 210096, China
        (email: jinshi@seu.edu.cn).}
         }


\maketitle
\vspace{-35pt}
\begin{abstract}
In this paper, we investigate the fundamental multiple access (MA) scheme in an active intelligent reflecting surface (IRS) aided energy-constrained Internet-of-Things (IoT) system, where an active IRS is deployed to assist the uplink transmission from multiple IoT devices to an access point (AP). Our goal is to maximize the sum throughput by optimizing the IRS beamforming vectors across time and resource allocation. To this end, we first study two typical active IRS aided MA schemes, namely time division multiple access (TDMA) and non-orthogonal multiple access (NOMA), by analytically comparing their achievable sum throughput and proposing corresponding algorithms. Interestingly, we prove that given only one available IRS beamforming vector, the NOMA-based scheme generally achieves a larger throughput than the TDMA-based scheme, whereas the latter can potentially outperform the former if multiple IRS beamforming vectors are available to harness the favorable time selectivity of the IRS. To strike a flexible balance between the system performance and the associated signaling overhead incurred by more IRS beamforming vectors, we then propose a general hybrid TDMA-NOMA scheme with user grouping, where the devices in the same group transmit simultaneously via NOMA while devices in different groups occupy orthogonal time slots. By controlling the number of groups, the hybrid TDMA-NOMA scheme is applicable for any given number of IRS beamforming vectors available. Despite of the non-convexity of the considered optimization problem, we propose an efficient algorithm based on alternating optimization, where each subproblem is solved optimally. Simulation results illustrate the practical superiorities of the active IRS over the passive IRS in terms of the coverage extension and supporting multiple energy-limited devices, and demonstrate the effectiveness of our proposed hybrid MA scheme for flexibly balancing the performance-cost tradeoff.
\end{abstract}



\IEEEpeerreviewmaketitle
\section{Introduction}
With the rapid development of Internet-of-Thighs (IoT) technologies, the unprecedented proliferation of electronic tablets, wearable devices, and mobile sensors is set to continue, which hastens a variety of IoT applications such as smart transportation, smart metering, and smart cities \cite{8879484}. In a typical IoT system, multiple sensor devices connect with an access point (AP) to form a wireless data collection system, which has been widely deployed in various practical applications, i.e., event detection for emergency services, external environment monitoring, wireless surveillance for public safety, etc., \cite{lin2017survey}. For these applications, the low-cost devices gather a large volume data from specific areas and then send it to the AP for further processing. However, the limited battery capacity of IoT devices is one of critical issues due to their practical size and cost constraints, which fundamentally limits their information uploading capabilities \cite{7516608, chi2019energy}.

To overcome the aforementioned limitations, new cost-effective wireless techniques have to be developed for assisting data transmission of IoT systems. Recently, intelligent reflecting surface (IRS) has emerged as a promising technology for enhancing the spectral efficiency and energy efficient of next generation wireless networks in a cost-effective way \cite{8910627, 9326394, 9424177}. Generally, IRS technologies mainly involve two types of architectures, namely passive IRS and active IRS. In particular, passive IRS is a digitally-controlled meta-surface composed of a large number of low-cost passive elements in tuning the phase shifts of the incident signals. With the proper design of the phase shifts of each element, IRS is capable of enhancing the signal reception at the desired destinations and/or mitigating the interference to unintended users, thereby artificially establishing favorable propagation conditions without requiring any RF chains. The fundamental power scaling law of passive IRS was firstly unveiled in \cite{8811733}, which demonstrated that passive IRS can provide an asymptotic squared-power gain in terms of received power at users via passive beamforming. The above advantages of passive IRS have then inspired an intensive research interest in optimizing IRS phase shifts for different wireless communication system setups, such as multi-cell cooperation \cite{pan2020multicell, 9279253, xie2020max}, physical layer security \cite{yu2020robust, guan2020intelligent, hu2021robust},  millimeter-wave communications \cite{lu2020robust, wang2021joint}, and unmanned aerial vehicle communications \cite{wei2020sum, mu2021intelligent}. While these works aimed at exploiting passive IRS for enhancing wireless information transmissions (WIT) of cellular networks, its high passive beamforming gain is also practically appealing for multifarious IoT application scenarios and unlocking its full potential in extending the lifetime of energy-constrained IoT devices \cite{wu2021intelligent}. Regarding the IRS-enabled IoT systems, several works have been emerged on three typical research lines, namely IRS-aided wireless information and power transfer (SWIPT) \cite{8941080, 9257429, 9133435}, IRS-aided wireless powered communication networks (WPCNs) \cite{9214497, 9298890, wu2021irs, 9650755}, and IRS-aided mobile edge computing (MEC) \cite{hu2021reconfigurable, 9270605, chen2021irs}.


Specifically, all of the aforementioned works have considered to exploit passive IRS for assisting wireless communications. Nevertheless, one critical issue of passive IRS-aided wireless systems is that its performance may be practically limited by the well-known high (product-distance) path-loss \cite{9326394}. To address the issue of passive IRS, a new type of IRS, called active IRS, has been recently proposed in \cite{long2021active, zhang2021active, you2021wireless, zeng2021throughput}, by amplifying the incident signals with low-cost hardware. Different from passive IRS, active IRS comprises a number of active elements, which are equipped with low-cost negative resistance components (e.g., negative impedance converter and tunnel diode), thereby enabling to amplify the reflected signals \cite{long2021active}. It is worth noting that the active IRS is quite different from the amplify-and-forward (AF) relay although it is capable of amplifying incident signals. Specifically, for the conventional AF relay, power-consuming RF chains are generally needed to receive signals first and then transmit it with the power amplification. In contrast, the active IRS does not need RF chains and the basic operation mechanism of the active IRS is similar to that of the passive IRS, which directly amplify/reflect signals in the air with low-power reflection-type amplifiers. By leveraging the active IRS in wireless transmissions of cellular networks, the joint AP and IRS beamforming design problems were investigated in \cite{long2021active, zhang2021active} for different system setups, i.e., single user uplink and multi-user downlink communication systems, respectively. In addition, the optimization of the active IRS deployment was studied in \cite{you2021wireless}. Benefited by smartly controlling the amplification gain at each element, the results in \cite{long2021active, zhang2021active, you2021wireless} demonstrated that the active IRS can perform better than the passive IRS in most practical scenarios.

Despite the aforementioned advantages of applying active IRS for assisting cellular networks, the investigations on the employment of active IRS in energy-constrained IoT systems are still in its infancy. Note that massive connectivity is another requirement of IoT systems, which thus calls for efficient multiple access (MA) schemes to support large number of IoT devices. Non-orthogonal multiple access (NOMA) is practically appealing for IoT networks due to its capability to enable the access of massive numbers of devices by allowing multiple users to simultaneously access the same spectrum \cite{8114722}. To the best of authors' knowledge, there is no existing work that investigates the potential performance gain of employment active IRSs in energy-constrained IoT systems by considering different MA schemes. Regarding an active IRS-aided energy-constrained IoT system, several fundamental issues remain unsolved. First, does time division multiple access (TDMA) outperform NOMA in such systems? Since active IRS is able to proactively establish favorable time-varing wireless channels, it is generally believed that exploiting dedicated IRS beamforming patterns for each individual device has a beneficial effect for TDMA. Additionally, the amplification gain of each element at active IRS is highly dependent on the transmit power of devices. For the typical energy-limited scenario, the transmit power of each device when employing NOMA is generally lower than that of TDMA for a given amount of energy, which renders that NOMA may reap larger available amplification gains at the active IRS compared to TDMA. Taking the above factors into consideration, it still remains an open problem which MA scheme is more beneficial for maximizing the system throughput. Second, how to design a more advanced MA scheme that is capable of flexibly striking a balance between performance and signaling overhead? This question is driven by the fact that even if TDMA outperforms NOMA, it also incurs higher signalling overhead since more IRS beamforming patterns are needed. As such, it may not be preferable to rely on the pure TDMA-based scheme considering the performance-cost tradeoff, especially when the number of IoT devices is practically large.

Motivated by the above issues, we study an active IRS-aided energy-constrained IoT system considering different MA schemes, where an active IRS is deployed to assist the UL data transmission between multiple energy-constrained devices and an AP.
Different from the conventional passive IRS, the ability of power amplification for the active IRS provides new degrees of freedom to combat against the severe double path loss and further enhance the received signal power. On the other hand, it also introduces new optimization variables and makes the IRS beamforming vectors and transmit power of each device closely coupled in the newly added IRS amplification power constraints, thus rendering the joint design of the IRS beamforming and resource allocation more challenging than that of the conventional passive IRS. The main contributions of this paper are summarized as follows.

\begin{itemize}
  \item We first study active IRS-aided energy-constrained IoT systems by utilizing both TDMA and NOMA schemes. For the TDMA scheme, the IRS beamforming vectors can be adjusted dynamically across time for each individual device, whereas for the NOMA scheme, all the devices share the same set of IRS beamforming vector during their data transmission. By utilizing the proposed models, we formulate the corresponding system sum throughput maximization problems by jointly optimizing the transmit power of each device, time allocation, and IRS beamforming vectors, subject to the energy constraints of IoT devices and the IRS amplification power constraints.
  \item For the TDMA scheme, we exploit the inherent properties of the optimal solution and propose an efficient algorithm based on successive convex approximation (SCA) to solve its associated optimization problem, where all the variables are optimized simultaneously. For the NOMA scheme, we propose an alternating optimization (AO)-based method to partition the entire variables into two blocks, namely the transmit power of devices and IRS beamforming vectors. Based on semidefinite program (SDP) techniques and Charnes-Cooper transformations, each block of variables is obtained optimally in an iterative way until convergence is achieved.
  \item Regarding the achievable sum throughput of the active IRS aided TDMA and NOMA schemes, we prove that given only one available IRS beamforming vector, the NOMA based scheme generally achieves a larger throughput than the TDMA based scheme, whereas the latter can potentially outperform the former if multiple IRS beamforming vectors are available. To provide more flexibility for balancing the performance-cost tradeoff, we propose a hybrid TDMA-NOMA scheme, where multiple devices are partitioned into several groups and the devices in the same group transmit simultaneously via NOMA while devices in different groups occupy orthogonal time resources. The proposed scheme generalizes the TDMA and NOMA schemes as two special cases and is applicable for any given number of IRS beamforming vectors available. We further extend the AO-based method to solve its associated optimization problem by applying proper change of variables.
  \item Our numerical results validate the theoretical findings and demonstrate that the practical superiorities of the active IRS over the conventional passive IRS in terms of supporting multiple low-energy devices, extending coverage range, and reducing the required number of reflecting elements. Moreover, it is shown that our proposed hybrid TDMA-NOMA scheme is capable of significantly lowering the signaling overhead at the cost of slight performance loss by properly determining the number of devices in each group.
\end{itemize}

The rest of this paper is organized as follows. Section II presents the system model for the active IRS-aided energy-constrained IoT system and problem formulations considering TDMA and NOMA. Sections III introduces proposed efficient algorithms for the corresponding problems in Section II and provides the theoretical performance comparison for the TDMA and NOMA-based schemes. In Section IV, we propose a general hybrid TDMA-NOMA scheme and extend the AO-based algorithm for solving its associated optimization problem. Section V presents numerical results to evaluate the performance of our proposed schemes and draw useful insights. Finally, we conclude the paper in Section VI.

\emph{Notations:} Boldface upper-case and lower-case  letter denote matrix and   vector, respectively.  ${\mathbb C}^ {d_1\times d_2}$ stands for the set of  complex $d_1\times d_2$  matrices. For a complex-valued vector $\bf x$, ${\left\| {\bf x} \right\|}$ represents the  Euclidean norm of $\bf x$, ${\rm arg}({\bf x})$ denotes  the phase of   $\bf x$, and ${\rm diag}(\bf x) $ denotes a diagonal matrix whose main diagonal elements are extracted from vector $\bf x$.
For a vector $\bf x$, ${\bf x}^*$ and  ${\bf x}^H$  stand for  its conjugate and  conjugate transpose respectively.   For a square matrix $\bf X$,  ${\rm{Tr}}\left( {\bf{X}} \right)$, $\left\| {\bf{X}} \right\|_2$ and ${\rm{rank}}\left( {\bf{X}} \right)$ respectively  stand for  its trace, Euclidean norm and rank,  while ${\bf{X}} \succeq {\bf{0}}$ indicates that matrix $\bf X$ is positive semi-definite.
A circularly symmetric complex Gaussian random variable $x$ with mean $ \mu$ and variance  $ \sigma^2$ is denoted by ${x} \sim {\cal CN}\left( {{{\mu }},{{\sigma^2 }}} \right)$.  ${\cal O}\left(  \cdot  \right)$ is the big-O computational complexity notation.
\begin{figure}[!t]
\setlength{\abovecaptionskip}{-5pt}
\setlength{\belowcaptionskip}{-5pt}
\centering
\includegraphics[width= 0.5\textwidth]{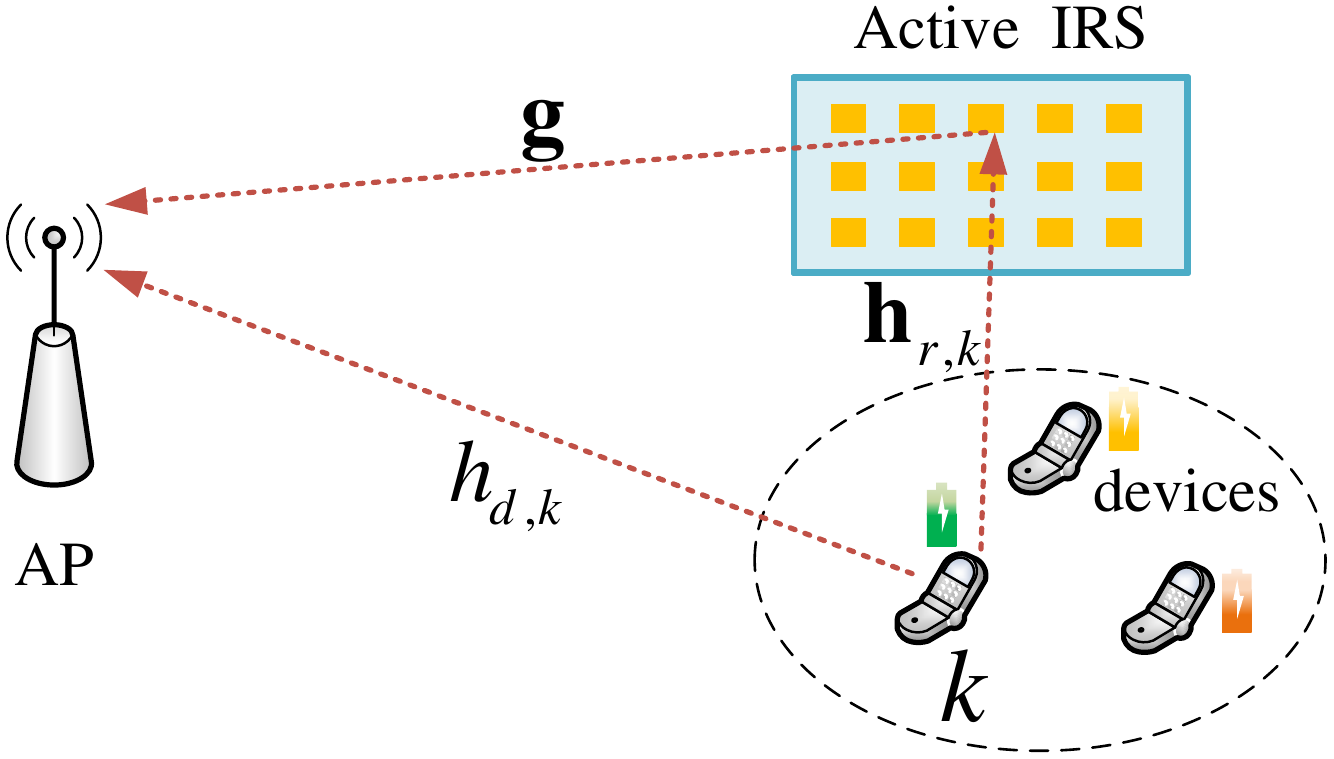}
\DeclareGraphicsExtensions.
\caption{An active IRS-aided energy-constrained uplink communication system.}
\label{model}
\vspace{-12pt}
\end{figure}
\section{System Models and Problem Formulations}

\subsection{System Model}
As shown in Fig. 1, we consider an active IRS-aided energy-constrained uplink communication system, where an active IRS with $N$ elements, denoted by the set, ${\cal N} \buildrel \Delta \over = \left\{ {1, \ldots ,N} \right\}$, is deployed to improve the efficiency of data transmission from $K$ single-antenna IoT devices, denoted by the set ${\cal K} \buildrel \Delta \over = \left\{ {1, \ldots ,K} \right\}$, to a  single-antenna AP. Without loss of generality, we focus on a typical energy-constrained IoT scenario, where a certain amount of energy ${E_k}$ Joule (J) for decice $k$ is available at the beginning of each transmission period. For example, the energy sources for each device can be ambient energy sources (e.g., TV signals or a renewable energy source) or dedicated WPT.
Then, all the devices use their available energy to transmit their own data to the AP in the uplink{\footnote{Note that our considered model generalizes various IoT applications. For example, it can be applicable to an MEC application, where energy-constrained devices offload their computation-intensive tasks to MEC servers integrated at an AP, or a wireless sensor network, where devices upload sensed data to the AP for further processing based on their available limited energy.}}. In addition, all devices and the AP are assumed to operate over the same frequency band and the time duration of each transmission period is denoted by ${T_{\max }}$. Furthermore, let ${\bf{g}} \in \mathbb{C}^{N\times 1}$, ${{\bf{h}}_{r,k}} \in \mathbb{C}^{N \times 1}$, and ${{h}_{d,k}} \in \mathbb{C}$ denote the equivalent baseband channels from the IRS to the AP, from device $k$ to the IRS, and from device $k$ to the AP, respectively. All the wireless channels are assumed to be quasi-static flat-fading and thereby remain constant within each transmission period ${T_{\max }}$. To characterize the fundamental performance comparisons for different MA schemes, we assume that the instantaneous CSI for all links is available based on the compressive sensing technique or the active sensors deployed at the IRS \cite{9326394}.

\subsection{TDMA and NOMA-based Multiple Access}
For a typical transmission period, each device can transmit its own information signal to the AP. Furthermore, we propose two data transmission setups depending on whether TDMA or NOMA is used, as detailed below.

\subsubsection{TDMA-based Scheme}
For the TDMA-based scheme, the AP receives information signals from different devices, which occupy orthogonal time slots (TSs). Let ${\tau _k}$ denote the time duration of the $k$-th TS, which is allocated for device $k$. Thus, we have $\sum\nolimits_{k = 1}^K {{\tau _k} \le {T_{\max }}}$. In the $k$-th TS, a dedicated IRS beamforming pattern, denoted by ${{\bf{\Phi }}_k} = {\rm{diag}}\left( {{\phi _{k,1}}, \ldots ,{\phi _{k,N}}} \right)$, is employed to reflect and amplify the transmitted signals. In particular, the reflecting/amplfication coefficient of the $n$-th element is denoted by ${\phi _{k,n}} = {a_{k,n}}{e^{j{\theta _{k,n}}}},n \in {\cal N}$, where ${a_{k,n}}$ and ${{\theta _{k,n}}}$, ${\theta _{k,n}} \in \left[ {0,2\pi } \right)$, represent the corresponding amplitude and phase. It is worth noting that ${a_{k,n}}$ can be greater than 1 with active loads \cite{long2021active}. We denote the transmit power and the information bearing signal of device $k$ as ${p_{{k}}}$ and ${s_{{k}}}$, respectively, which satisfies ${\mathop{\rm E}\nolimits} \left\{ {{{\left| {{s_{{k}}}} \right|}^2}} \right\} = 1$. In the $k$-th TS, the signal reflected and amplified by the active IRS is given by
\begin{align}\label{reflect_signal}
{{\bf{y}}_{r,k}} = {{\bf{\Phi }}_k}{{\bf{h}}_{r,k}}\sqrt {{p_k}} {s_k} + {{\bf{\Phi }}_k}{{\bf{n}}_r},\;\;k \in {\cal K},
\end{align}
where ${{\bf{n}}_r} \in \mathbb{C}^{N \times 1}$ represents the thermal noise generated at the active IRS, which is distributed as ${\cal C}{\cal N}\left( {0,\sigma _r^2{{\bf{I}}_N}} \right)$. Note that the active IRS amplifies both the incident signal and noise. We denote the maximum amplification power of the active IRS as ${P_r}$ and thus we have
\begin{align}\label{power_constraint}
{\rm{E}}\left[ {{{\left\| {{{\bf{y}}_{r,k}}} \right\|}^2}} \right] = {p_k}{\left\| {{{\bf{\Phi }}_k}{{\bf{h}}_{r,k}}} \right\|^2} + \sigma _r^2\left\| {{{\bf{\Phi }}_k}} \right\|_F^2 \le {P_r},\;\;k \in {\cal K}.
\end{align}
Additionally, the received signal at the AP in the $k$-th TS can be expressed as
\begin{align}\label{received_signal}
{{\bf{y}}_k} = \left( {{h_{d,{k_l}}} + {{\bf{g}}^H}{{\bf{\Phi }}_k}{{\bf{h}}_{r,k}}} \right)\sqrt {{p_k}} {s_k} + {{\bf{g}}^H}{{\bf{\Phi }}_k}{{\bf{n}}_r} + n,\;\;k \in {\cal K},
\end{align}
where ${n} \sim {\cal C}{\cal N}\left( {0,1} \right)$ denotes the additive white Gaussian noise at the AP. As such, the sum achievable throughput of the system in bits/Hz is given by
\begin{align}\label{throughput_TDMA}
{R_{{\rm{TDMA}}}}=\sum\limits_{k = 1}^K {{\tau _k}{{\log }_2}\left( {1 + \frac{{{p_k}{{\left| {{{{h}}_{d,{k}}} + {{\bf{g}}^H}{{\bf{\Phi }}_k}{{\bf{h}}_{r,k}}} \right|}^2}}}{{{\sigma ^2} + \sigma _r^2{{\left\| {{{\bf{g}}^H}{{\bf{\Phi }}_k}} \right\|}^2}}}} \right)}.
\end{align}

\subsubsection{NOMA-based Scheme}
For the NOMA-based scheme, all the devices transmit simultaneously to the AP. As such, a common IRS beamforming pattern, denoted by ${\bf{\Phi }}$, is shared by all devices. In this case, the signal reflected and amplified by the active IRS is given by
\begin{align}\label{reflect_signal_NOMA}
{{\bf{y}}_r} = {\bf{\Phi }}\sum\limits_{k = 1}^K {{{\bf{h}}_{r,k}}\sqrt {{p_k}} {s_k}}  + {\bf{\Phi }}{{\bf{n}}_r}.
\end{align}
Correspondingly, we have the following IRS amplification power constraint for the NOMA scheme
\begin{align}\label{power_constraint_NOMA}
{\rm{E}}\left[ {{{\left\| {{{\bf{y}}_r}} \right\|}^2}} \right] = \sum\limits_{k = 1}^K {{p_k}{{\left\| {{\bf{\Phi }}{{\bf{h}}_{r,k}}} \right\|}^2}}  + \sigma _r^2\left\| {\bf{\Phi }} \right\|_F^2 \le {P_r}.
\end{align}
The received signal at the AP can be further expressed as
\begin{align}\label{received_signal}
{\bf{y}} = \sum\limits_{k = 1}^K {\left( {{h_{d,k}} + {{\bf{g}}^H}{\bf{\Phi }}{{\bf{h}}_{r,k}}} \right)\sqrt {{p_k}} {s_k}}  + {{\bf{g}}^H}{\bf{\Phi }}{{\bf{n}}_r} + n.
\end{align}
To mitigate the multiuser interference, successive interference cancellation is performed at the AP. Taking device $k$ as an example, the AP will first decode the message of device $i$, $\forall i < k$, before decoding the message of device $k$. Then, the message of device $i$, $\forall i < k$, will be subtracted from the composite signal. The message received from device $i$, $\forall i > k$, is treated as noise. Thus, the achievable sum throughput of all devices can be written as \cite{8114722}
\begin{align}\label{throughput_NOMA}
{R_{{\rm{NOMA}}}}=\tau {\log _2}\left( {1 + \sum\nolimits_{k = 1}^K {\frac{{{p_k}{{\left| {{{{h}}_{d,{k}}} + {{\bf{g}}^H}{\bf{\Phi }}{{\bf{h}}_{r,k}}} \right|}^2}}}{{{\sigma ^2} + \sigma _r^2{{\left\| {{{\bf{g}}^H}{\bf{\Phi }}} \right\|}^2}}}} } \right),
\end{align}
where $\tau$ denotes the transmission time duration for all devices and $\tau  \le {T_{\max }}$.

\subsection{Problem Formulation}
We aim to maximize the sum throughput of the considered system by jointly optimizing the time allocation, the transmit power of each device, and the active IRS beamforming. For the TDMA-based scheme, the optimization problem can be expressed as
\begin{subequations}\label{C5}
\begin{align}
\label{C5-a}\mathop {\max }\limits_{\left\{ {{\tau _k}} \right\},\left\{ {{p_k}} \right\},\left\{ {{{\bf{\Phi }}_k}} \right\}}  \;\;&\sum\limits_{k = 1}^K {{\tau _k}{{\log }_2}\left( {1 + \frac{{{p_k}{{\left| {{{{h}}_{d,{k}}} + {{\bf{g}}^H}{{\bf{\Phi }}_k}{{\bf{h}}_{r,k}}} \right|}^2}}}{{{\sigma ^2} + \sigma _r^2{{\left\| {{{\bf{g}}^H}{{\bf{\Phi }}_k}} \right\|}^2}}}} \right)}\\
\label{C5-b}{\rm{s.t.}}\;\;\;\;\;\;\;&{\tau _k}{p_k} \le {E_k}, ~\forall {k} \in {{\cal K}},\\
\label{C5-c}&\sum\limits_{k = 1}^K {{\tau _k} \le {T_{\max }}},\\
\label{C5-d}&{\tau _k} \ge 0, ~{p_{{k}}} \ge 0, ~\forall {k} \in {{\cal K}},\\
\label{C5-e}&{p_k}{\left\| {{{\bf{\Phi }}_k}{{\bf{h}}_{r,k}}} \right\|^2} + \sigma _r^2\left\| {{{\bf{\Phi }}_k}} \right\|_F^2 \le {P_r},~~\forall {k} \in {{\cal K}}.
\end{align}
\end{subequations}
For problem \eqref{C5}, constraint \eqref{C5-b} ensures that the total energy consumed at each device cannot exceed its available energy. Constraints \eqref{C5-c} and \eqref{C5-d} are the total transmission time constraint and the non-negative constraints on the optimization variables, respectively, while constraint \eqref{C5-e} indicates that the amplification power of the active IRS should not exceed the maximum allowed power. Similarly, the optimization problem associated with the NOMA-based scheme can be formulated as
\begin{subequations}\label{C6}
\begin{align}
\label{C6-a}\mathop {\max }\limits_{{\tau },\left\{ {{p_k}} \right\},{{\bf{\Phi }}}}  \;\;&\tau {\log _2}\left( {1 + \sum\limits_{k = 1}^K {\frac{{{p_k}{{\left| {{{{h}}_{d,{k}}} + {{\bf{g}}^H}{\bf{\Phi }}{{\bf{h}}_{r,k}}} \right|}^2}}}{{{\sigma ^2} + \sigma _r^2{{\left\| {{{\bf{g}}^H}{\bf{\Phi }}} \right\|}^2}}}} } \right)\\
\label{C6-b}{\rm{s.t.}}\;\;\;\;&\tau {p_k} \le {E_k}, ~\forall {k} \in {{\cal K}},\\
\label{C6-c}&\tau  \le T,\\
\label{C6-d}&\tau \ge 0, ~{p_{{k}}} \ge 0, ~\forall {k} \in {{\cal K}},\\
\label{C6-e}&\sum\limits_{k = 1}^K {{p_k}{{\left\| {{\bf{\Phi }}{{\bf{h}}_{r,k}}} \right\|}^2}}  + \sigma _r^2\left\| {\bf{\Phi }} \right\|_F^2 \le {P_r}.
\end{align}
\end{subequations}

The above two problems \eqref{C5} and \eqref{C6} are all non-convex since the optimization variables are closely coupled in the objective function and constraints. Therefore, there are no standard methods for solving such non-convex optimization problems optimally in general. Moreover, it is worth noting that the hidden structures of problems \eqref{C5} and \eqref{C6} are fundamentally different. Specifically, in \eqref{C6-a} and \eqref{C6-e}, the common IRS beamforming pattern, i.e., ${\bf{\Phi }}$, is coupled with all ${p_k}$'s, while the transmit power of device $k$, i.e., ${p_k}$, is only coupled with its individual IRS beamforming pattern, i.e., ${{\bf{\Phi }}_k}$, in \eqref{C5-a} and \eqref{C5-e}. These issues have non-trivial effects on the algorithms design and performance comparison for the optimal values of problems \eqref{C5} and \eqref{C6}. In the next section, by deeply exploiting inherent properties of the corresponding optimization problems, we propose efficient algorithms to obtain high-quality solutions for them and provide a fundamental performance comparison for the active IRS aided TDMA and NOMA schemes.

\section{Investigations on Active IRS Aided TDMA and NOMA Schemes}
In this section, we study the active IRS aided TDMA and NOMA schemes. By exploiting the specific structures of the associated problems, two dedicated algorithms are proposed for problems \eqref{C5} and \eqref{C6}, respectively. Furhermore, we provide a theoretical performance comparison for the achievable sum throughput of the active IRS aided TDMA and NOMA schemes.
\subsection{Proposed Algorithm for Active IRS Aided TDMA}
For the TDMA-based scheme, the sum throughput maximization problem can be rewritten in a more tractable form as
\begin{subequations}\label{C7}
\begin{align}
\label{C7-a}\mathop {\max }\limits_{\left\{ {{\tau _k}} \right\},\left\{ {{p_k}} \right\},\left\{ {{{\bf{v}}_k}} \right\}}  \;\;&\sum\limits_{k = 1}^K {{\tau _k}{{\log }_2}\left( {1 + \frac{{{p_k}{{\left| {{h_{d,k}} + {\bf{v}}_k^H{{\bf{q}}_k}} \right|}^2}}}{{{\sigma ^2} + \sigma _r^2{\bf{v}}_k^H{\bf{G}}{{\bf{v}}_k}}}} \right)}\\
\label{C7-b}{\rm{s.t.}}\;\;\;\;\;\;\;\;&\eqref{C5-b}, \eqref{C5-c}, \eqref{C5-d},\\
\label{C7-e}&{p_k}{\bf{v}}_k^H{{\bf{H}}_{r,k}}{{\bf{v}}_k} + \sigma _r^2{\left\| {{{\bf{v}}_k}} \right\|^2} \le {P_r},~~\forall {k} \in {{\cal K}},
\end{align}
\end{subequations}
where ${\bf{v}}_k^H = \left[ {{\phi _{k,1}}, \ldots ,{\phi _{k,N}}} \right]$, ${{\bf{q}}_k} = {\mathop{\rm diag}\nolimits} \left( {{{\bf{g}}^H}} \right){{\bf{h}}_{r,k}}$, ${\bf{G}} = {\mathop{\rm diag}\nolimits} \left( {{{\left| {{{\left[ {\bf{g}} \right]}_1}} \right|}^2}, \ldots {{\left| {{{\left[ {\bf{g}} \right]}_N}} \right|}^2}} \right)$, and ${{\bf{H}}_{r,k}} = {\mathop{\rm diag}\nolimits} \left( {{{\left| {{{\left[ {{{\bf{h}}_{r,k}}} \right]}_1}} \right|}^2}, \ldots {{\left| {{{\left[ {{{\bf{h}}_{r,k}}} \right]}_N}} \right|}^2}} \right)$. Generally, problem \eqref{C7} is a non-convex optimization problem and difficult to solve optimally due to the non-concave objective function as well as coupled variables in constraints \eqref{C5-b} and \eqref{C7-e}. Different from the conventional passive IRS-aided communication systems, the new IRS amplification power constraint, i.e., constraint \eqref{C7-e}, is involved. As a result, the value of the amplification coefficient at each element for the active IRS may be reduced as the transmit power of each device, i.e., ${p_k}$, increases, which may weaken the effectiveness of the active IRS and thus has negative effects on the sum throughput. Therefore, it is not clear whether the energy of each device would be used up for maximizing the sum throughput, i.e., \eqref{C5-b} is active or not at the optimal solution, which motivates the following lemma.
\begin{lem}
At the optimal solution of problem \eqref{C7}, constraint \eqref{C7-e} is always met with equality, i.e., ${\tau _k}{p_k} = {E_k}$.
\end{lem}
\begin{proof}
Please refer to Appendix A.
\end{proof}

Lemma 1 reveals that for the active IRS-aided energy-constrained uplink communication system with TDMA, each device will deplete all of its energy at the optimal solution, i.e., constraint \eqref{C7-b} holds with equality. Thus, problem \eqref{C7} is equivalently simplified to the following
\begin{subequations}\label{C8}
\begin{align}
\label{C8-a}\mathop {\max }\limits_{\left\{ {{\tau _k}} \right\},\left\{ {{{\bf{v}}_k}} \right\}}  \;\;&\sum\limits_{k = 1}^K {{\tau _k}{{\log }_2}\left( {1 + \frac{{{E_k}{{\left| {{h_{d,k}} + {\bf{v}}_k^H{{\bf{q}}_k}} \right|}^2}}}{{{\tau _k}\left( {{\sigma ^2} + \sigma _r^2{\bf{v}}_k^H{\bf{G}}{{\bf{v}}_k}} \right)}}} \right)}\\
\label{C8-b}{\rm{s.t.}}\;\;\;\;\;&{\tau _k} \ge 0, ~\forall {k} \in {{\cal K}},\\\
\label{C8-d}&\frac{{{E_k}}}{{{\tau _k}}}{\bf{v}}_k^H{{\bf{H}}_{r,k}}{{\bf{v}}_k} + \sigma _r^2{\left\| {{{\bf{v}}_k}} \right\|^2} \le {P_r},~~\forall {k} \in {{\cal K}}.\\
\label{C8-e}&\eqref{C5-c}.
\end{align}
\end{subequations}
Note that by exploiting Lemma 1, constraint \eqref{C7-e} is transformed into \eqref{C8-d}, which is a convex constraint since ${E_k}$ is a constant. The remaining challenge for solving problem \eqref{C8} is the non-concave objective function. Nevertheless, we propose an efficient algorithm to solve it, where all the variables are optimized simultaneously. First, we introduce a set of slack variables denoted by $\left\{ {{S_k}} \right\}$ and reformulate problem \eqref{C8} as
\begin{subequations}\label{C9}
\begin{align}
\label{C9-a}\mathop {\max }\limits_{\left\{ {{\tau _k}} \right\},\left\{ {{{\bf{v}}_k}} \right\},\left\{ {{S_k}} \right\}}  \;\;&\sum\limits_{k = 1}^K {{\tau _k}{{\log }_2}\left( {1 + \frac{{{S_k}}}{{{\tau _k}}}} \right)}\\
\label{C9-b}{\rm{s.t.}}\;\;\;\;\;\;\;\;&{S_k} \le \frac{{{E_k}{{\left| {{h_{d,k}} + {\bf{v}}_k^H{{\bf{q}}_k}} \right|}^2}}}{{{\sigma ^2} + \sigma _r^2{\bf{v}}_k^H{\bf{G}}{{\bf{v}}_k}}}, ~\forall {k} \in {{\cal K}},\\
\label{C9-d}&\eqref{C5-c}, \eqref{C8-b}, \eqref{C8-d}.
\end{align}
\end{subequations}
For the optimal solution of problem \eqref{C9}, constraint \eqref{C9-b} is met with equality strictly since we can always increase the objective value by increasing ${S_k}$ until constraint \eqref{C9-b} becomes active. Thus, problem \eqref{C9} is equivalent to problem \eqref{C8}. However, constraint \eqref{C9-b} is still non-convex. To deal with constraint \eqref{C9-b}, we rewrite it into a more tractable form as
\begin{align}\label{constraint_form}
{\sigma ^2} + \sigma _r^2{\bf{v}}_k^H{\bf{G}}{{\bf{v}}_k} \le \frac{{{E_k}{{\left| {{h_{d,k}} + {\bf{v}}_k^H{{\bf{q}}_k}} \right|}^2}}}{{{S_k}}}, \forall {k} \in {{\cal K}}.
\end{align}
It is observed that ${\sigma ^2} + \sigma _r^2{\bf{v}}_k^H{\bf{G}}{{\bf{v}}_k}$ is a convex quadratic function of ${{\bf{v}}_k}$ while the right-hand-side of \eqref{constraint_form} is jointly convex with respect to ${{\bf{v}}_k}$ and ${{S_k}}$. Recall that any convex function is globally lower-bounded by its first-order Taylor expansion at any feasible point, which motivates us to employ the SCA technique to deal with the non-convexity of \eqref{constraint_form}. Specifically, for the given local point $\left\{ {{\bf{v}}_k^{\left( l \right)},S_k^{\left( l \right)}} \right\}$ in the $l$-th iteration, we have the following lower bound as
\begin{align}\label{SCA_bound1}
\frac{{{E_k}{{\left| {{h_{d,k}} + {\bf{v}}_k^H{{\bf{q}}_k}} \right|}^2}}}{{{S_k}}} &\ge \frac{{2{E_k}}}{{S_k^{\left( l \right)}}}{\mathop{\rm Re}\nolimits} \left( {{{\left( {{h_{d,k}} + {\bf{v}}_k^H{{\bf{q}}_k}} \right)}^H}\left( {{h_{d,k}} + {{\left( {{\bf{v}}_k^{\left( l \right)}} \right)}^H}{{\bf{q}}_k}} \right)} \right)\nonumber\\
&~~~~ - \frac{{{E_k}{{\left| {{h_{d,k}} + {{\left( {{\bf{v}}_k^{\left( l \right)}} \right)}^H}{{\bf{q}}_k}} \right|}^2}}}{{{{\left( {S_k^{\left( l \right)}} \right)}^2}}}{S_k}\buildrel \Delta \over = f_k^{lb}\left( {{{\bf{v}}_k},{S_k}} \right).
\end{align}
It can be readily checked that $f_k^{lb}\left( {{{\bf{v}}_k},{S_k}} \right)$ is linear and convex with respect to $\left\{ {{{\bf{v}}_k},{S_k}} \right\}$. As such, with the lower bound in \eqref{SCA_bound1}, constraint \eqref{C9-b} is transformed to
\begin{align}\label{approximate_constraint}
{\sigma ^2} + \sigma _r^2{\bf{v}}_k^H{\bf{G}}{{\bf{v}}_k} \le f_k^{lb}\left( {{{\bf{v}}_k},{S_k}} \right), \forall {k} \in {{\cal K}}.
\end{align}
Then, problem \eqref{C9} can be approximated by
\begin{subequations}\label{C10}
\begin{align}
\label{C10-a}\mathop {\max }\limits_{\left\{ {{\tau _k}} \right\},\left\{ {{{\bf{v}}_k}} \right\},\left\{ {{S_k}} \right\}}  \;\;&\sum\limits_{k = 1}^K {{\tau _k}{{\log }_2}\left( {1 + \frac{{{S_k}}}{{{\tau _k}}}} \right)}\\
\label{C0-b}{\rm{s.t.}}\;\;\;\;\;\;\;\;&\eqref{C5-c}, \eqref{C8-d}, \eqref{C8-e}, \eqref{approximate_constraint},
\end{align}
\end{subequations}
which is a convex optimization problem. Thus, we can apply existing standard convex optimization tools to successively solve it optimally until the convergence is achieved. After convergence, it can be guaranteed that a locally optimal solution for original problem \eqref{C7} can be obtained. The computational complexity of this algorithm lies in solving problem \eqref{C10} and is given by ${\cal O}\left( {{{\left( {K\left( {N + 2} \right)} \right)}^{3.5}}{I_{{\rm{Iter}}}}} \right)$, where ${{I_{{\rm{Iter}}}}}$ represents the number of iterations required for convergence.

\subsection{Proposed Algorithm for NOMA}
For the active IRS-aided NOMA scheme, the sum throughput maximization problem can be rewritten as
\begin{subequations}\label{C11}
\begin{align}
\label{C11-a}\mathop {\max }\limits_{{\tau},\left\{ {{p_k}} \right\},{{\bf{v}}}}  \;\;&\tau {\log _2}\left( {1 + \frac{{\sum\nolimits_{k = 1}^K {{p_k}{{\left| {{h_{d,k}} + {{\bf{v}}^H}{{\bf{q}}_k}} \right|}^2}} }}{{{\sigma ^2} + \sigma _r^2{{\bf{v}}^H}{\bf{Gv}}}}} \right)\\
\label{C11-b}{\rm{s.t.}}\;\;\;\;&\eqref{C6-b}, \eqref{C6-c}, \eqref{C6-d}\\
\label{C11-e}&\sum\limits_{k = 1}^K {{p_k}{{\bf{v}}^H}{{\bf{H}}_{r,k}}{\bf{v}}}  + \sigma _r^2{\left\| {\bf{v}} \right\|^2} \le {P_r},
\end{align}
\end{subequations}
where ${{\bf{v}}^H} = \left[ {{\phi _1}, \ldots ,{\phi _N}} \right]$. Different from problem \eqref{C7}, all devices share the same IRS beamforming vector in problem \eqref{C11} and the common IRS beamforming vector ${\bf{v}}$ is coupled with the transmit power of all devices in \eqref{C11-a} and \eqref{C11-e}. Thus, ${\bf{v}}$ cannot be flexibly adjusted for each individual device and the amplification gains of the active IRS may be locked if all devices transmit at their maximum allowed power simultaneously. Furthermore, different from the TDMA case, some devices may not use up all of their energy at the optimal solution for the NOMA case.
As such, the algorithm proposed for problem \eqref{C7} is not applicable to the more challenging problem \eqref{C11}, which thus calls for the new algorithm design. Problem \eqref{C11} is generally intractable due to the non-concave objective function and non-convex constraints \eqref{C6-b} and \eqref{C11-e}. To tackle the coupled variables ${p_k}$ and $\tau$ in constraint \eqref{C6-b}, we have the following lemma.
\begin{lem}
At the optimal solution of problem \eqref{C11}, constraint \eqref{C6-c} is strictly met with equality, i.e., $\tau  = {T_{\max }}$.
\end{lem}
\begin{proof}
We firstly define ${e_k} = {p_k}\tau$ and then the objective function of problem \eqref{C11} can be rewritten as
\begin{align}\label{obj_function}
\tau {\log _2}\left( {1 + \frac{{\sum\nolimits_{k = 1}^K {{e_k}{{\left| {{h_{d,k}} + {{\bf{v}}^H}{{\bf{q}}_k}} \right|}^2}} }}{{\tau \left( {{\sigma ^2} + \sigma _r^2{{\bf{v}}^H}{\bf{Gv}}} \right)}}} \right).
\end{align}
Suppose ${\tau ^*}$ is the optimal transmission time. Then, we show that ${\tau ^*} = {T_{\max }}$ by contradiction as follows. Assume that ${\Xi ^*} = \left\{ {{\tau ^*},p_k^*,{{\bf{v}}^*}} \right\}$ achieves the optimal solution of problem \eqref{C11} and ${\tau ^*} < {T_{\max }}$. Then, we construct a different solution $\tilde \Xi  = \left\{ {\tilde \tau ,{{\tilde p}_k},{\bf{\tilde v}}} \right\}$, where ${\tau ^*} < \tilde \tau  \le {T_{\max }}$, $e_k^* = \tilde \tau {{\tilde p}_k} = {\tau ^*}p_k^*$, and ${\bf{\tilde v}} = {{\bf{v}}^*}$. It can be readily verified that $\left\{ {{{\tilde p}_k},{\bf{\tilde v}}} \right\}$ satisfies constraint \eqref{C11-e} since ${{\tilde p}_k} \le p_k^*$. As such, the constructed solution satisfies all the constraints in problem \eqref{C11}. Since \eqref{obj_function} is an increasing function with respect to $\tau$ and $\tilde \tau  > {\tau ^*}$, we have
\begin{align}\label{obj_function_compare}
\tilde \tau {\log _2}\left( {1 + \frac{{\sum\nolimits_{k = 1}^K {e_k^*{{\left| {{h_{d,k}} + {{{\bf{\tilde v}}}^H}{{\bf{q}}_k}} \right|}^2}} }}{{\tilde \tau \left( {{\sigma ^2} + \sigma _r^2{{{\bf{\tilde v}}}^H}{\bf{G\tilde v}}} \right)}}} \right) > {\tau ^*}{\log _2}\left( {1 + \frac{{\sum\nolimits_{k = 1}^K {e_k^*{{\left| {{h_{d,k}} + {{\bf{v}}^*}^H{{\bf{q}}_k}} \right|}^2}} }}{{{\tau ^*}\left( {{\sigma ^2} + \sigma _r^2{{\bf{v}}^*}^H{\bf{G}}{{\bf{v}}^*}} \right)}}} \right).
\end{align}
\eqref{obj_function_compare} indicates that the constructed solution ${\tilde \Xi }$ achieves a higher objective value, which contradicts that the optimal ${\tau ^*} < {T_{\max }}$. This thus completes the proof.
\end{proof}

Exploiting Lemma 2, problem \eqref{C11} can be simplified to
\begin{subequations}\label{C12}
\begin{align}
\label{C12-a}\mathop {\max }\limits_{\left\{ {{p_k}} \right\},{{\bf{v}}}}  \;\;&{T_{\max }}{\log _2}\left( {1 + \frac{{\sum\nolimits_{k = 1}^K {{p_k}{{\left| {{h_{d,k}} + {{\bf{v}}^H}{{\bf{q}}_k}} \right|}^2}} }}{{{\sigma ^2} + \sigma _r^2{{\bf{v}}^H}{\bf{Gv}}}}} \right)\\
\label{C12-b}{\rm{s.t.}}\;\;\;&{p_k} \le \frac{{{E_k}}}{{{T_{\max }}}}, ~\forall {k} \in {{\cal K}},\\
\label{C12-c}&{p_{{k}}} \ge 0, ~\forall {k} \in {{\cal K}},\\
\label{C12-d}&\eqref{C11-e}.
\end{align}
\end{subequations}

The main challenges for solving problem \eqref{C12} are the non-concave objective function \eqref{C12-a} and the coupled variables $\left\{ {{p_k},{\bf{v}}} \right\}$ involved in constraint \eqref{C12-d}, which motivates us to apply the AO-based method to solve it. Specifically, we divide all the variables into two blocks, i.e., 1) IRS beamforming vector ${\bf{v}}$, and 2) power control ${{p_k}}$, and then each block of variables is optimized in an iterative way, until convergence is achieved.
\subsubsection{IRS Beamforming Optimization}
For any given transmit power ${{p_k}}$, the IRS beamforming vector optimization problem is given by
\begin{subequations}\label{C14}
\begin{align}
\label{C14-a}\mathop {\max }\limits_{{{\bf{v}}}}  \;\;&{T_{\max }}{\log _2}\left( {1 + \frac{{\sum\nolimits_{k = 1}^K {{p_k}{{\left| {{h_{d,k}} + {{\bf{v}}^H}{{\bf{q}}_k}} \right|}^2}} }}{{{\sigma ^2} + \sigma _r^2{{\bf{v}}^H}{\bf{Gv}}}}} \right)\\
\label{C14-b}{\rm{s.t.}}\;\;&\eqref{C12-d}.
\end{align}
\end{subequations}
Let ${\left| {{h_{d,k}} + {{\bf{v}}^H}{{\bf{q}}_k}} \right|^2} = {\left| {{{{\bf{\bar v}}}^H}{{{\bf{\bar q}}}_k}} \right|^2}$, where ${{{\bf{\bar v}}}^H} = \left[ {{{\bf{v}}^H},1} \right]$ and ${{{\bf{\bar q}}}_k} = {\left[ {{\bf{q}}_k^H,h_{d,k}^H} \right]^H}$. Define ${{{\bf{\bar Q}}}_k} = {{{\bf{\bar q}}}_k}{\bf{\bar q}}_k^H$, ${\bf{\bar V}} = {\bf{\bar v}}{{{\bf{\bar v}}}^H}$, which needs to satisfy ${\bf{\bar V}} \succeq {\bf{0}}$, ${\rm{rank}}\left( {{\bf{\bar V}}} \right) = 1$ and ${\left[ {{\bf{\bar V}}} \right]_{N + 1,N + 1}} = 1$. We thus have
\begin{align}\label{intended_signal}
\sum\limits_{k = 1}^K {{p_k}{{\left| {{h_{d,k}} + {{\bf{v}}^H}{{\bf{q}}_k}} \right|}^2} = } \sum\limits_{k = 1}^K {{p_k}{\mathop{\rm Tr}\nolimits} } \left( {{{{\bf{\bar Q}}}_k}{\bf{\bar V}}} \right) = {\rm{Tr}}\left( {{\bf{\bar Q\bar V}}} \right),
\end{align}
where ${\bf{\bar Q}} = \sum\nolimits_{k = 1}^K {{p_k}} {{{\bf{\bar Q}}}_k}$. Define ${\bf{\bar G}} = {\mathop{\rm diag}\nolimits} \left( {{\bf{G}},0} \right)$, ${{{\bf{\bar H}}}_{r,k}} = {\mathop{\rm diag}\nolimits} \left( {{{\bf{H}}_{r,k}},0} \right)$, and ${{{\bf{\bar H}}}_r} = \sum\limits_{k = 1}^K {{p_k}{{{\bf{\bar H}}}_{r,k}}}  + \sigma _r^2{{\bf{\Pi }}_r}$, where ${{\bf{\Pi }}_r} = {\mathop{\rm diag}\nolimits} \left( {1, \ldots 1,0} \right)$. We further have
\begin{align}\label{noise}
{\sigma ^2} + \sigma _r^2{{\bf{v}}^H}{\bf{Gv}} = {\sigma ^2} + \sigma _r^2{\rm{Tr}}\left( {{\bf{\bar G\bar V}}} \right),\\
{p_k}{{\bf{v}}^H}{{\bf{H}}_{r,k}}{\bf{v}} + \sigma _r^2{\left\| {\bf{v}} \right\|^2} = {\rm{Tr}}\left( {{{{\bf{\bar H}}}_r}{\bf{\bar V}}} \right).
\end{align}
As such, we can reformulate problem \eqref{C14} in an equivalent form as follows
\begin{subequations}\label{C15}
\begin{align}
\label{C15-a}\mathop {\max }\limits_{{{{\bf{\bar V}}}}}  \;\;&\frac{{{\rm{Tr}}\left( {{\bf{\bar Q\bar V}}} \right)}}{{{\sigma ^2} + \sigma _r^2{\rm{Tr}}\left( {{\bf{\bar G\bar V}}} \right)}}\\
\label{C15-b}{\rm{s.t.}}\;\;&{\rm{Tr}}\left( {{{{\bf{\bar H}}}_r}{\bf{\bar V}}} \right) \le {P_r}, \\
\label{C15-c}&{\bf{\bar V}} \succeq {\bf{0}},\\
\label{C15-d}&{\left[ {{\bf{\bar V}}} \right]_{N + 1,N + 1}} = 1,\\
\label{C15-e}&{\rm{rank}}\left( {{\bf{\bar V}}} \right) = 1.
\end{align}
\end{subequations}
Problem \eqref{C15} is still intractable due to the fraction form in the objective function and the rank-one constraint. To tackle this issue, we relax the rank-one constraint and apply the Charnes-Cooper transformation to reformulate it as a linear form as
\begin{subequations}\label{C16}
\begin{align}
\label{C16-a}\mathop {\max }\limits_{{{{\bf{\bar V}}}}, t}  \;\;&{{\rm{Tr}}\left( {{\bf{\bar Q\bar V}}} \right)}\\
\label{C16-b}{\rm{s.t.}}\;\;&\sigma _r^2{\rm{Tr}}\left( {{\bf{\bar G\bar V}}} \right) + t{\sigma ^2} = 1, \\
\label{C16-c}&{\rm{Tr}}\left( {{{{\bf{\bar H}}}_r}{\bf{\bar V}}} \right) \le t{P_r},\\
\label{C16-d}&t > 0,\\
\label{C16-e}&{\left[ {{\bf{\bar V}}} \right]_{N + 1,N + 1}} = t,\\
\label{C16-f}&\eqref{C15-c}.
\end{align}
\end{subequations}
\begin{lem}
By relaxing the rank-one constraint (37e) in problem \eqref{C15}, problem \eqref{C15} is equivalent to problem \eqref{C16}.
\end{lem}
\begin{proof}
First, given any feasible solution $\left\{ {{\bf{\bar V}}} \right\}$ to problem \eqref{C15}, it can be verified that with the solution $\left\{ {{{{\bf{\bar V}}} \mathord{\left/
 {\vphantom {{{\bf{\bar V}}} {\left( {{\sigma ^2} + \sigma _r^2{\rm{Tr}}\left( {{\bf{\bar G\bar V}}} \right)} \right),{1 \mathord{\left/
 {\vphantom {1 {\left( {{\sigma ^2} + \sigma _r^2{\rm{Tr}}\left( {{\bf{\bar G\bar V}}} \right)} \right)}}} \right.
 \kern-\nulldelimiterspace} {\left( {{\sigma ^2} + \sigma _r^2{\rm{Tr}}\left( {{\bf{\bar G\bar V}}} \right)} \right)}}}}} \right.
 \kern-\nulldelimiterspace} {\left( {{\sigma ^2} + \sigma _r^2{\rm{Tr}}\left( {{\bf{\bar G\bar V}}} \right)} \right),{1 \mathord{\left/
 {\vphantom {1 {\left( {{\sigma ^2} + \sigma _r^2{\rm{Tr}}\left( {{\bf{\bar G\bar V}}} \right)} \right)}}} \right.
 \kern-\nulldelimiterspace} {\left( {{\sigma ^2} + \sigma _r^2{\rm{Tr}}\left( {{\bf{\bar G\bar V}}} \right)} \right)}}}}} \right\}$, problem \eqref{C16} achieves the same objective value as that of \eqref{C15}. Then, given any feasible solution $\left\{ {{\bf{\bar V}},t} \right\}$ to problem \eqref{C16}, it can be similarly demonstrated that with the solution $\left\{ {{{{\bf{\bar V}}} \mathord{\left/
 {\vphantom {{{\bf{\bar V}}} t}} \right.
 \kern-\nulldelimiterspace} t}} \right\}$, problem \eqref{C15} achieves the same objective value. As such, problem \eqref{C15} and \eqref{C16} have the same optimal value. Lemma 3 is thus proved.
\end{proof}

It can be readily verified that problem \eqref{C16} is a convex optimization problem, whose optimal solution can be efficiently solved by the standard convex optimization techniques. According to Lemma 3, we can solve problem \eqref{C16} instead of solving \eqref{C15}.
\bibliographystyle{IEEEtran}
\begin{rem}
According to Theorem 3.2 in \cite{huang2009rank}, we can conclude that there always exists a optimal solution ${{{\bf{\bar V}}}^*}$ to problem \eqref{C16}, satisfies the following constraint:
\begin{align}\label{noise}
{\rm{Ran}}{{\rm{k}}^2}\left( {{{{\bf{\bar V}}}^*}} \right) + {\rm{Ran}}{{\rm{k}}^2}\left( {{t^*}} \right) \le 3.
\end{align}
We note that ${\rm{Ran}}{{\rm{k}}^2}\left( {{t^*}} \right) = 1$. Thus, we have ${\rm{Rank}}\left( {{{{\bf{\bar V}}}^*}} \right) = 1$. There always exists rank-one solution for problem \eqref{C16}.
\end{rem}

Based on Remark 1, we can obtain the optimal rank-one solution for problem \eqref{C16}, denoted by $\left\{ {{{{\bf{\bar V}}}^*},{t^*}} \right\}$. Thus, the optimal solution for problem \eqref{C15} is obtained as ${{{{{\bf{\bar V}}}^*}} \mathord{\left/
 {\vphantom {{{{{\bf{\bar V}}}^*}} {{t^*}}}} \right.
 \kern-\nulldelimiterspace} {{t^*}}}$. By performing singular value decomposition (SVD) for ${{{{{\bf{\bar V}}}^*}} \mathord{\left/
 {\vphantom {{{{{\bf{\bar V}}}^*}} {{t^*}}}} \right.
 \kern-\nulldelimiterspace} {{t^*}}}$, the optimal solution ${{\bf{v}}^*}$ can be found for original problem \eqref{C14}.
\subsubsection{Power Control Optimization}
For any given IRS beamforming vector ${\bf{v}}$, the power control optimization problem is given by
\begin{subequations}\label{C17}
\begin{align}
\label{C17-a}\mathop {\max }\limits_{\left\{ {{p_k}} \right\}}  \;\;&{\frac{{\sum\nolimits_{k = 1}^K {{p_k}{{\left| {{h_{d,k}} + {{\bf{v}}^H}{{\bf{q}}_k}} \right|}^2}} }}{{{\sigma ^2} + \sigma _r^2{{\bf{v}}^H}{\bf{Gv}}}}}\\
\label{C17-b}{\rm{s.t.}}\;\;&\eqref{C12-b}, \eqref{C12-c}, \eqref{C12-d}.
\end{align}
\end{subequations}
Since the objective function and all constraints in \eqref{C17} are linear, problem \eqref{C17} is a convex optimization problem, which can be optimally solved by the standard convex optimization methods.
\subsubsection{Overall Algorithm and Computational Complexity Analysis}
Based on the solutions to the above two subproblems, an efficient AO algorithm is proposed, where the IRS beamforming vector and power control are alternately optimized until convergence is achieved. Note that the objective value of problem \eqref{C12} is non-decreasing by alternately optimizing $\left\{ {\bf{v}} \right\}$ and $\left\{ {{p_k}} \right\}$, thus the proposed AO algorithm is guaranteed to converge. The mainly computational complexity of this AO algorithm lies from solving problems \eqref{C16} and \eqref{C17}. Specifically, the computational complexity for solving problems \eqref{C16} and \eqref{C17} is given by ${\cal O}\left( {{{\left( {N + 1} \right)}^{3.5}}} \right)$ and ${\cal O}\left( {{K^{3.5}}} \right)$, respectively. Therefore, the total complexity of the AO algorithm is ${\cal O}\left( {\left( {{{\left( {N + 1} \right)}^{3.5}} + {K^{3.5}}} \right){I_{{\rm{AO}}}}} \right)$, where ${{I_{{\rm{AO}}}}}$ denotes the number of iterations required to reach convergence.

\subsection{Active IRS Aided TDMA Versus NOMA}
Compared to TDMA, NOMA is expected to achieve a higher throughput in a conventional communication system by allowing multiple devices simultaneously to access the same spectrum. In our considered active IRS aided energy-constrained IoT systems, each device can occupy its dedicated IRS beamforming under the TDMA-based scheme, while all the devices share the same IRS beamforming vector under the NOMA-based scheme. As such, it is not clear which MA scheme can achieve higher throughput. In this subsection, we give some discussions about the performance comparison for the active IRS aided TDMA and NOMA schemes.

Firstly, we introduce a special case of the active IRS aided TDMA where only one beamforming vector is available for assisting uplink transmission, which leads to the following problem formulation
\begin{subequations}\label{C18}
\begin{align}
\label{C18-a}\mathop {\max }\limits_{\left\{ {{\tau _k}} \right\},\left\{ {{p_k}} \right\},{\bf{v}}}  \;\;&\sum\limits_{k = 1}^K {{\tau _k}{{\log }_2}\left( {1 + \frac{{{p_k}{{\left| {{h_{d,k}} + {{\bf{v}}^H}{{\bf{q}}_k}} \right|}^2}}}{{{\sigma ^2} + \sigma _r^2{{\bf{v}}^H}{\bf{Gv}}}}} \right)}\\
\label{C18-b}{\rm{s.t.}}\;\;\;\;\;\;&\eqref{C5-b}, \eqref{C5-c}, \eqref{C5-d},\\
\label{C18-c}&{p_k}{\bf{v}}^H{{\bf{H}}_{r,k}}{{\bf{v}}} + \sigma _r^2{\left\| {{{\bf{v}}}} \right\|^2} \le {P_r},~~\forall {k} \in {{\cal K}}.
\end{align}
\end{subequations}
Intuitively, it seems that constraint \eqref{C11-e} is tighter than \eqref{C18-c} in terms of the transmit power. Thus, it may be expected that the optimal value of problem \eqref{C18} is no smaller than that of \eqref{C11}. However, the result is counterintuitive as shown in the following proposition.
\begin{proposition}
Denote the optimal values of problem \eqref{C11} and \eqref{C18} by $R_{{\rm{NOMA}}}^*$ and $R_{{\rm{TDMA}}}^{\left( {lb} \right)*}$, respectively, it follows that
\begin{align}\label{TDMA and NOMA}
R_{{\rm{NOMA}}}^* \ge R_{{\rm{TDMA}}}^{\left( {lb} \right)*}.
\end{align}

If
\begin{align}\label{condition3}
{{\tilde p}_k}{\left| {{h_{d,k}} + {{{\bf{\tilde v}}}^H}{{\bf{q}}_k}} \right|^2} \ne {{\tilde p}_j}{\left| {{h_{d,j}} + {{{\bf{\tilde v}}}^H}{{\bf{q}}_j}} \right|^2},\exists k,j \in {\cal K},
\end{align}
$R_{{\rm{NOMA}}}^* > R_{{\rm{TDMA}}}^{\left( {{\rm{lb}}} \right)*}$ holds, where $\left\{ {{{\tilde p}_k},{\bf{\tilde v}}} \right\}$ is the optimal solution of problem \eqref{C18}.

Moreover, a sufficient condition for $R_{{\rm{NOMA}}}^* = R_{{\rm{TDMA}}}^{\left( {lb} \right)*}$ is
\begin{align}\label{sufficient_condition}
{{\mathord{\buildrel{\lower3pt\hbox{$\scriptscriptstyle\smile$}}
\over p} }_k}{{{\bf{\mathord{\buildrel{\lower3pt\hbox{$\scriptscriptstyle\smile$}}
\over v} }}}^H}{{\bf{H}}_{r,k}}{\bf{\mathord{\buildrel{\lower3pt\hbox{$\scriptscriptstyle\smile$}}
\over v} }} + \sigma _r^2{\left\| {{\bf{\mathord{\buildrel{\lower3pt\hbox{$\scriptscriptstyle\smile$}}
\over v} }}} \right\|^2} \le {P_r},\forall k \in {\cal K},
\end{align}
where $\left\{ {{{\mathord{\buildrel{\lower3pt\hbox{$\scriptscriptstyle\smile$}}
\over p} }_k},{\bf{\mathord{\buildrel{\lower3pt\hbox{$\scriptscriptstyle\smile$}}
\over v} }}} \right\}$ is the optimal solution of problem \eqref{C18} when relaxing constraint \eqref{C18-c}.
\end{proposition}
\begin{proof}
Please refer to Appendix B.
\end{proof}

Proposition 1 implies that the active IRS aided NOMA scheme achieves a no smaller sum throughput than its counterpart of the TDMA scheme when only one beamforming vector is available. This seems contradictory to the conclusion of the previous work \cite{chen2021irs} in the passive IRS aided energy-constrained scenario. The reason is that for an energy-constrained system, the transmit power of each device under the NOMA case is generally lower than that of the TDMA case. This makes NOMA be able to reap larger amplification gains of the active IRS compared to TDMA, which compensates the performance loss induced by the low transmit power. Based on Proposition 1, we give a sufficient condition for that the active IRS aided TDMA scheme can achieve a higher throughput than NOMA in the following theorem.
\begin{thm}
Denote the optimal value of problem \eqref{C7} by $R_{{\rm{TDMA}}}^*$. Then, \eqref{sufficient_condition} serves as a sufficient condition for $R_{{\rm{TDMA}}}^* \ge R_{{\rm{NOMA}}}^*$.
\end{thm}
\begin{proof}
According to the results in Proposition 1, we have $R_{{\rm{TDMA}}}^{\left( {{\rm{lb}}} \right)*} = R_{{\rm{NOMA}}}^*$ when \eqref{sufficient_condition} is satisfied. Note that problem \eqref{C18} is a special case of problem \eqref{C7} with ${{\bf{v}}_k} = {{\bf{v}}_j},\forall k,j \in {\cal K}$. As such, the optimal solution of problem \eqref{C18} is also one feasible solution of problem \eqref{C7}, which yields $R_{{\rm{TDMA}}}^* \ge R_{{\rm{TDMA}}}^{\left( {{\rm{lb}}} \right)*} = R_{{\rm{NOMA}}}^*$. This thus completes the proof.
\end{proof}
\begin{rem}
Theorem 1 implies that the active IRS aided TDMA scheme generally outperforms its counterpart with NOMA due to more IRS beamforming vectors can be exploited for TDMA. The above two MA schemes strike a balance between the throughput performance and the number of optimization variables as well as the feedback signalling overhead. Specifically, the TDMA-based scheme requires the AP to optimize and feedback $KN$ IRS reflection coefficients (including both amplitudes and phase shifts) to the IRS, which increases linearly with the number of devices, while these required for the NOMA-based scheme is $N$. Therefore, the signaling overhead for TDMA in such an overloaded scenario is much heavier than that of NOMA. This motivates us to propose a more flexible scheme to fully harness the maximum performance gain of the active IRS with controllable signaling overhead and complexity, elaborated in the next section.
\end{rem}

\section{An Overhead-Aware Hybrid TDMA-NOMA Scheme}
Motivated by the discussions in the previous section, we develop a more general hybrid TDMA-NOMA-based uplink transmission scheme in this section for fully exploiting the tradeoff between the throughput performance and signalling overhead.
\vspace{-8pt}
\subsection{Problem Formulation for General Hybrid TDMA-NOMA Scheme}
\begin{figure}[!t]
\setlength{\abovecaptionskip}{-5pt}
\setlength{\belowcaptionskip}{-5pt}
\centering
\includegraphics[width= 0.75\textwidth]{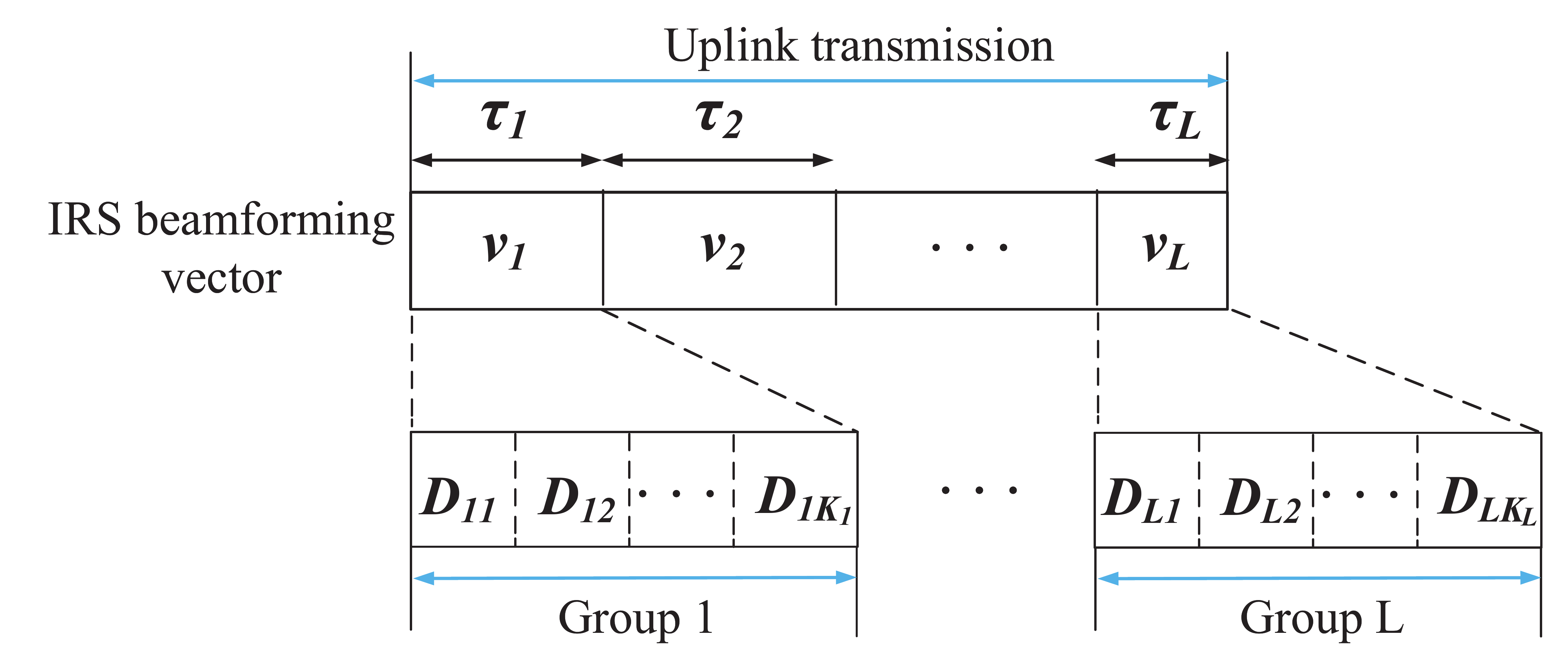}
\DeclareGraphicsExtensions.
\caption{The proposed general TDMA-NOMA scheme for active IRS-aided uplink communications.}
\label{model2}
\vspace{-20pt}
\end{figure}
The detailed transmission protocol is illustrated in Fig. 2. Without loss of generality, we assume that the IRS beamforming vectors can be reconfigured $L-1$ times during the transmission period, corresponding to $L$ available IRS beamforming vectors, i.e., ${{\bf{v}}_l}, l \in {\cal L} = \left\{ {1, \ldots L} \right\}$. We further partition $K$ devices into $L$ disjoint groups equally and the set of devices in $l$-th group is denoted by ${{\cal K}_l}$ with size ${K_l} = {K \mathord{\left/
 {\vphantom {K L}} \right.
 \kern-\nulldelimiterspace} L}$. Furthermore, the devices in different groups transmit in different TSs, while the devices in the same group transmit simultaneously employing NOMA. Specifically, each device ${k_l}$, ${k_l} \in {{\cal K}_l}$, will transmit in the $l$-th TS with the aid of IRS beamforming vector ${{\bf{v}}_l}$ and the time duration for the $l$-th TS is denoted by ${\tau _l}$. Under the given device sets\footnote{Note that in practice, the device grouping may be affected by many factors. Nevertheless, in Section V, we evaluate the impact of different device grouping  methods on the performance in our considered systems.} ${{\cal K}_l},\forall l \in {\cal L}$, the system sum throughput maximization problem by jointly optimizing the IRS beamforming vectors, the transmit power, and the time allocation, can be formulated as
\begin{subequations}\label{C19}
\begin{align}
\label{C19-a}\mathop {\max }\limits_{\left\{ {{\tau _l}} \right\},{\left\{ p \right\}_{{k_l}}},\left\{ {{{\bf{v}}_l}} \right\}}  \;\;&\sum\limits_{l = 1}^L {{\tau _l}{{\log }_2}\left( {1 + \sum\nolimits_{{k_l} \in {{\cal K}_l}} {\frac{{{p_{{k_l}}}{{\left| {{h_{d,{k_l}}} + {\bf{v}}_l^H{{\bf{q}}_{{k_l}}}} \right|}^2}}}{{{\sigma ^2} + {\bf{v}}_l^H{\bf{G}}{{\bf{v}}_l}}}} } \right)}\\
\label{C19-b}{\rm{s.t.}}\;\;\;\;\;\;\;\;&{\tau _l}{p_{{k_l}}} \le {E_{{k_l}}}, ~\forall l \in {\cal L}, \forall {k_l} \in {{\cal K}_l},\\
\label{C19-c}&\sum\limits_{l = 1}^L {{\tau _l} \le {T_{\max }}},\\
\label{C19-d}&{\tau _l} \ge 0, ~{p_{{k_l}}} \ge 0, ~\forall l \in {\cal L}, ~\forall {k_l} \in {{\cal K}_l},\\
\label{C19-e}&\sum\nolimits_{{k_l} \in {{\cal K}_l}} {{p_{{k_l}}}{\bf{v}}_l^H{{\bf{H}}_{r,{k_l}}}{{\bf{v}}_l}}  + \sigma _r^2{\left\| {{{\bf{v}}_l}} \right\|^2} \le {P_r},~\forall {k_l} \in {{\cal K}_l}.
\end{align}
\end{subequations}

It is worth noting that problem \eqref{C19} is equivalent to \eqref{C5} when $L = K$. In addition, when $L = 1$, it can be verified that problem \eqref{C19} is equivalent to \eqref{C6}. As such, both standalone TDMA and NOMA-based schemes are special cases of the proposed hybrid TDMA-NOMA-based scheme. By controlling the number groups, i.e., $L$, the proposed hybrid MA scheme can be applicable for any given number of IRS beamforming vectors imposed by the practical systems and operate in an overhead-aware manner.
\subsection{Proposed Solution for Problem \eqref{C19}}
Note that the non-convex problem \eqref{C19} is more challenging to solve than problems \eqref{C7} and \eqref{C11}. Specifically, different from problem \eqref{C7} that each device can occupy a dedicated IRS beamforming vector or problem \eqref{C11} that all devices share the same IRS beamforming vector. As such, Lemma 1 and Lemma 2 can not be applied here to simplify the problem. To handle this issue, we extend the AO-based framework to solve such a problem by partitioning the entire optimization variables into two blocks, but with the different elements in each block, namely, $\left\{ {{\tau _l},{p_{{k_l}}}} \right\}$ and $\left\{ {{{\bf{v}}_l}} \right\}$, for updating alternatingly. Next, we respectively solve the above two blocks.
\subsubsection{Optimizing $\left\{ {{\tau _l},{p_{{k_l}}}} \right\}$}
For the given block $\left\{ {{{\bf{v}}_l}} \right\}$, the subproblem for optimizing $\left\{ {{\tau _l},{p_{{k_l}}}} \right\}$ is still a non-convex optimization problem due to the coupled variables ${\tau _l}$ and ${p_{{k_l}}}$ in constraint \eqref{C19-b} and the non-concave objective function. We apply a change of variables as ${e_{{k_l}}} = {\tau _l}{p_{{k_l}}}$ and rewrite the subproblem in an equivalent form as follows
\begin{subequations}\label{C22}
\begin{align}
\label{C22-a}\mathop {\max }\limits_{\left\{ {{\tau _l}} \right\},\left\{ {{e_{{k_l}}}} \right\}}  \;\;&\sum\limits_{l = 1}^L {{\tau _l}{{\log }_2}\left( {1 + \sum\nolimits_{{k_l} \in {K_l}} {\frac{{{e_{{k_l}}}{{\left| {{h_{d,{k_l}}} + {\bf{v}}_l^H{{\bf{q}}_{{k_l}}}} \right|}^2}}}{{{\tau _l}\left( {{\sigma ^2} + {\bf{v}}_l^H{\bf{G}}{{\bf{v}}_l}} \right)}}} } \right)}\\
\label{C22-b}{\rm{s.t.}}\;\;\;\;\;&{e_{{k_l}}} \le {E_{{k_l}}}, ~\forall l \in {\cal L}, \forall {k_l} \in {{\cal K}_l},\\
\label{C22-c}&{\tau _l} \ge 0, ~{e_{{k_l}}} \ge 0, ~\forall l \in {\cal L}, ~\forall {k_l} \in {{\cal K}_l},\\
\label{C22-d}&\sum\nolimits_{{k_l} \in {{\cal K}_l}} {{e_{{k_l}}}{\bf{v}}_l^H{{\bf{H}}_{r,{k_l}}}{{\bf{v}}_l}}  + \sigma _r^2{\tau _l}{\left\| {{{\bf{v}}_l}} \right\|^2} \le {P_r}{\tau _l},~\forall l \in {\cal L},\\
\label{C22-e}&\eqref{C19-c}.
\end{align}
\end{subequations}
It can be readily checked that objective function \eqref{C22-a} is concave and all constraints in problem \eqref{C22} are convex. Therefore,  problem \eqref{C22} is convex and its optimal solution can be efficiently obtained by using standard convex optimization solvers such as CVX.

\subsubsection{Optimizing $\left\{ {{{\bf{v}}_l}} \right\}$}
For the given $\left\{ {{\tau _l},{p_{{k_l}}}} \right\}$, it is observed that optimization variables with respect to different groups are separable in both the objective function and constraints. Thus, the resultant problem with respect to $\left\{ {{{\bf{v}}_l}} \right\}$ can be addressed by solving $L$ independent subproblems in parallel, each with only one single constraint. Specifically, for group $l$, the corresponding subproblem with respect to ${{{\bf{v}}_l}}$'s, $\forall l \in {\cal L}$, is reduced to
\begin{subequations}\label{C23}
\begin{align}
\label{C21-a}\mathop {\max }\limits_{{{\bf{v}}_l}}  \;\;&{\sum\limits_{{k_l} \in {{\cal K}_l}} {\frac{{{p_{{k_l}}}{{\left| {{h_{d,{k_l}}} + {\bf{v}}_l^H{{\bf{q}}_{{k_l}}}} \right|}^2}}}{{{\sigma ^2} + {\bf{v}}_l^H{\bf{G}}{{\bf{v}}_l}}}} }\\
\label{C21-b}{\rm{s.t.}}\;\;&\sum\nolimits_{{k_l} \in {{\cal K}_l}} {{p_{{k_l}}}{\bf{v}}_l^H{{\bf{H}}_{r,{k_l}}}{{\bf{v}}_l}}  + \sigma _r^2{\left\| {{{\bf{v}}_l}} \right\|^2} \le {P_r}.
\end{align}
\end{subequations}
Since problem \eqref{C23} has the same form as that of problem \eqref{C14}, it can be similarly solved optimally by applying the Charnes-Cooper transformation-based SDP technique and the details are thus omitted for brevity.

Similar to the discussions in Section III-B, the computational complexity of the AO algorithm for solving problem \eqref{C19} is given by ${\cal O}\left( {\left( {{L^{3.5}}{{\left( {N + 1} \right)}^{3.5}} + {{\left( {2L} \right)}^{3.5}}} \right){I_{{\rm{AO}}}}} \right)$.

\section{Numerical Results}
In this section, we provide numerical results to validate the effectiveness of the proposed schemes and to draw useful insights into active IRS-aided energy-constrained IoT systems. The AP and the active IRS are located at $\left( {0,0,0} \right)$ meter (m) and $\left( {{x_{{\rm{IRS}}}},0,4} \right)$ m, respectively, and the devices are uniformly and randomly distributed within a radius of 5 m centered at $\left( {{x_{\rm{D}}},0,4} \right)$ m. The path-loss exponents of both the IRS-device and AP-IRS links are set to 2.2, while those of the AP-device links are set to 3.4. In addition, we assume that the AP-device link, the AP-IRS link, and the IRS-device link follow Rayleigh fading. The signal attenuation at a reference distance of 1 m is set as 30 dB. The other parameters are set as follows: $N = 50$, $T = 0.1$ s, ${E_k} = {E_m},\forall m \ne k$, ${x_{{\rm{IRS}}}} = 0$, ${x_{\rm{D}}} = 30$ m, and ${\sigma ^2} = \sigma _r^2 = -75$ dBm.
\subsection{Performance Comparison for TDMA and NOMA-based Schemes}
\begin{figure}[t!]
\centering
\includegraphics[width=3in]{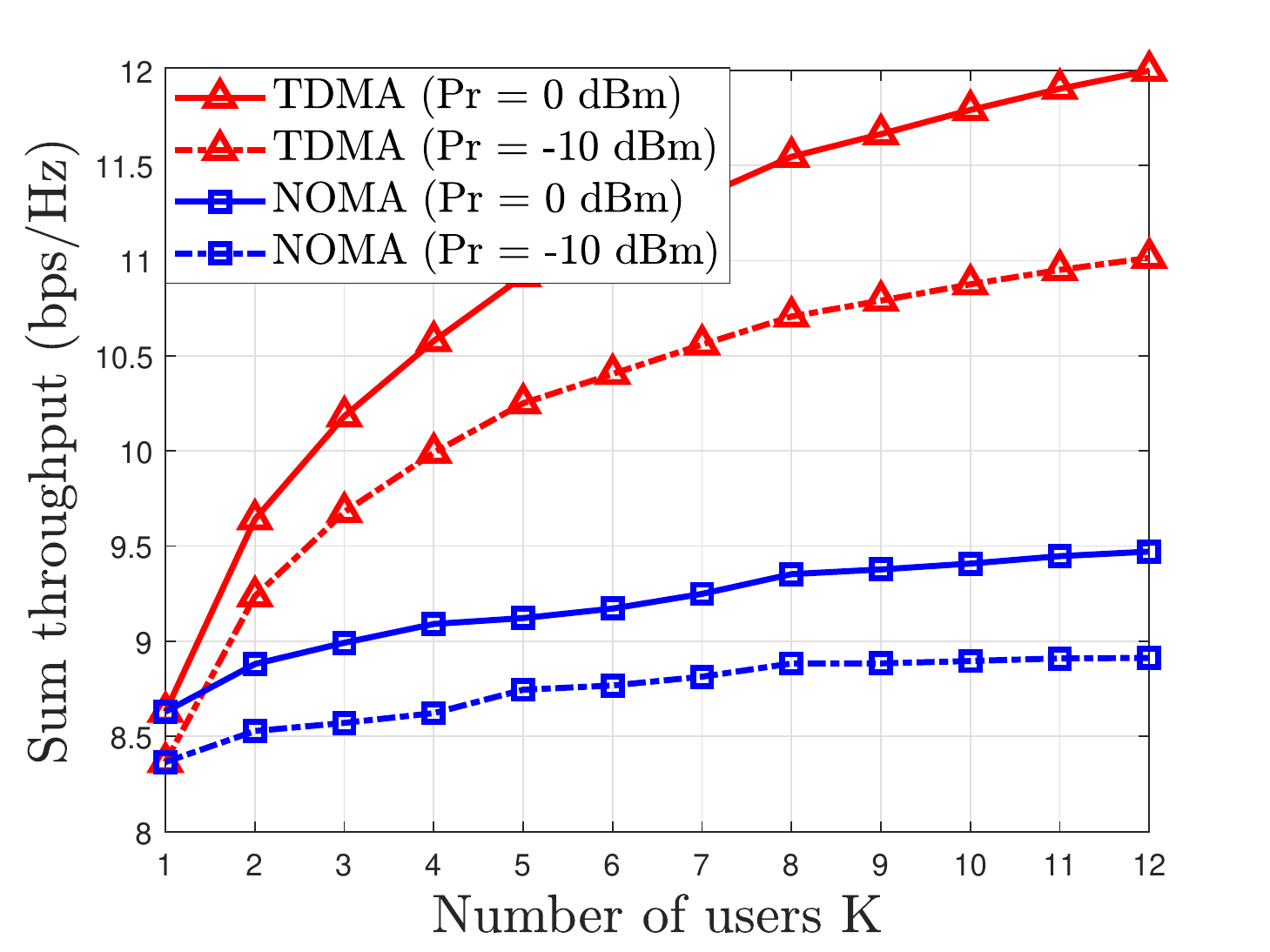}
\caption{{Sum throughput versus $K$ for active IRS aided TDMA and NOMA schemes. }}
\label{MA_comparison}
\vspace{-16pt}
\end{figure}

In this subsection, we compare two types of MA schemes, i.e., TDMA and NOMA, in terms of the sum throughput. In Fig. \ref{MA_comparison}, we study the impact of the number of devices $K$ on the system sum throughput for both ${P_r} = 0$ dBm and ${P_r} = -10$ dBm. For the TDMA-based scheme, we observe that the sum throughput increases as $K$ increases. This is expected since scheduling one more device would always result in a higher sum throughput, which verifies our analysis in Lemma 1. While for the NOMA-based scheme, the sum throughput increases slowly or remains at an almost constant value with the increase of $K$. Since all the devices share the common IRS beamforming vector for data transmission in the NOMA case, the amplification coefficients of the active IRS will be restricted due to \eqref{C6-e} as more devices are scheduled to transmit simultaneously. To fully unleash the potential amplifying capability of the active IRS, some devices may not be scheduled or transmit at its maximum power for maximizing the sum throughput, which renders that the sum throughput increases slowly with $K$ for the NOMA-based scheme.

Moreover, it is observed that for $K \ge 2$, the TDMA-based scheme always outperforms its NOMA-based counterpart and the performance gap becomes more pronounced with the increase of $K$. This is because the TDMA-based scheme can enjoy more degrees of freedom to flexibly adjust each device's dedicated IRS beamforming vector for further boosting the sum throughput compared to NOMA. As such, the TDMA-based scheme may be more preferable for the active IRS aided system at the cost of extra signalling overhead.

\subsection{Performance Comparison for Active IRS or Passive IRS}
\begin{figure}
\begin{minipage}[t]{0.45\linewidth}
\centering
\includegraphics[width=2.9in]{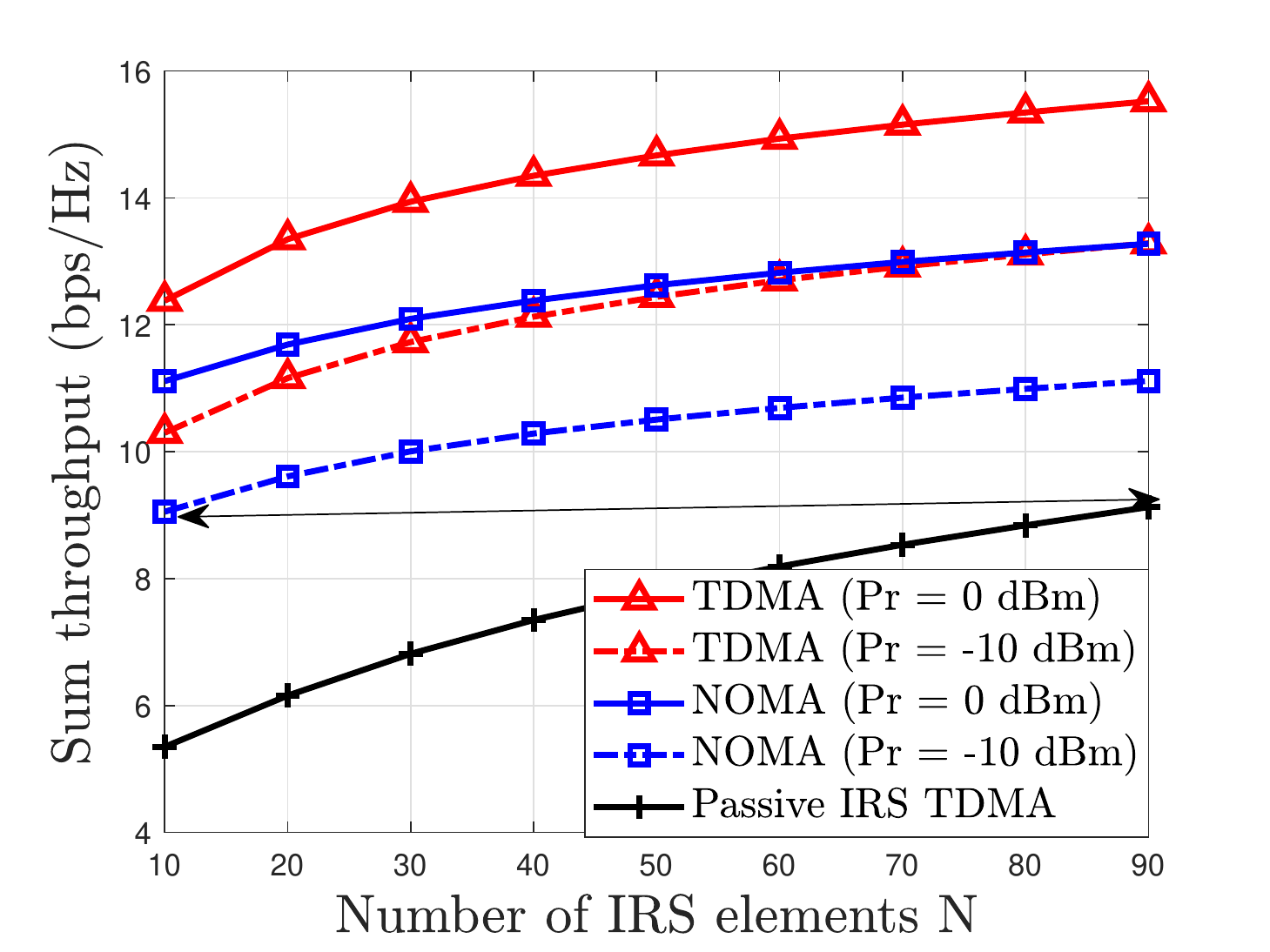}
\caption{Sum throughput versus the number of IRS elements with ${E_k} = 0.01{\rm{ J}}$ and $K = 10$.}
\label{element}
\end{minipage}%
\hfill
\begin{minipage}[t]{0.45\linewidth}
\centering
\includegraphics[width=2.9in]{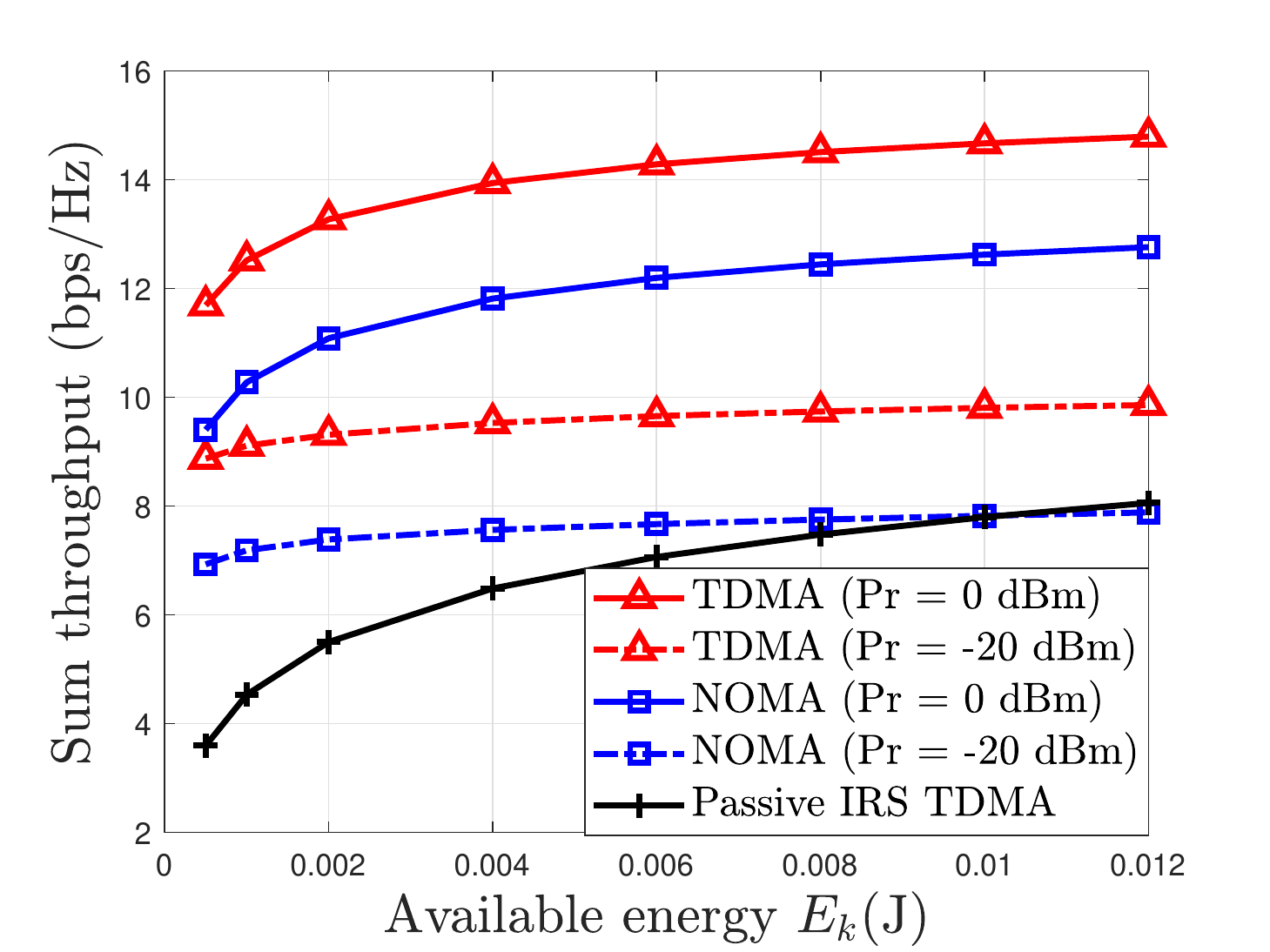}
\caption{Sum throughput versus the initial energy ${E_k}$ with $K = 10$.}
\label{energy}
\end{minipage}
\vspace{-16pt}
\end{figure}
In this subsection, we provide performance comparisons for the active IRS and passive IRS-aided energy-constrained IoT systems under different setups. For the passive IRS-aided system, we employ the algorithm proposed in \cite{wu2021irs} for obtaining the sum throughput of the TDMA-based scheme since it can characterize the upper bound performance of the passive IRS architecture in our considered scenarios.
\subsubsection{Impact of Number of IRS Elements}
In Fig. \ref{element}, we plot the system sum throughput versus the number of IRS elements $N$ for both ${P_r} = 0$ dBm and ${P_r} = -10$ dBm. It is observed that the sum throughput of all the schemes monotonically increases with $N$ since more reflecting/amplfication elements help achieve higher IRS beamforming gains, which is beneficial for improving the power of the received signals. Moreover, one can observe that even exploiting the active IRS for the NOMA-based transmission (i.e., one active beamforming vector) can achieve significant gains over employing passive IRS for the TDMA-based transmission (i.e., $K$ passive beamforming vectors) in terms of the system sum throughput. Considering the TDMA-based scheme relies on more IRS beamforming adjustments, the result indicates that incorporating the active IRS into IoT systems leads to a higher throughput with lower signalling overhead compared to the passive IRS. In addition, the performance gap between the two types of IRS architectures is more pronounced in the low $N$ regime. Note that to achieve 10 bps/Hz throughput, the required number of IRS elements is reduced from 90 to 10 via replacing the passive IRS by the active IRS, which indicates that the required number of IRS elements can be greatly reduced when employing the active IRS. The reason is that the value of the amplification coefficient at each element becomes larger when reducing the number of IRS elements. The higher amplification coefficient can effectively compensate the performance loss induced by reducing the number of IRS elements.
\subsubsection{Impact of Available Energy ${E_k}$}
In Fig. \ref{energy}, we study the impact of the available energy on our considered systems, by plotting the sum throughput versus the available energy at each device. From Fig. \ref{energy}, it is observed that the active IRS can significantly improve the sum throughput as compared to the case of the passive IRS, especially when the available energy at each device is low. This is because the power of the incident signal at each IRS element becomes smaller when the available energy at each device is low. As such, the maximum allowed amplification coefficient at each element of the active IRS becomes larger, which compensates the performance loss caused by the low available energy at devices. It implies that the active IRS is a more promising architecture for supporting multiple low-energy devices compared to the conventional passive IRS.
In addition, the sum throughput of active IRS-aided systems for the case of ${P_r} = -20$ dBm is less sensitive to ${E_k}$ compare to that of ${P_r} = 0$ dBm. The reason is that the lower value of ${P_r}$ would limit the amplification amplitude at each IRS element especially when the transmit power of each device is high.
\subsubsection{Impact of Distance Between AP and Device Center}
In Fig. \ref{coverage}, we investigate the coverage performance of our considered active IRS-aided systems, by plotting the sum throughput versus the distance between the AP and the center of devices cluster. It is expected that the sum throughput of all the considered schemes decreases as the distance increases. However, for both TDMA and NOMA cases, the decreasing rate of the active IRS-aided systems is much slower than that of the passive IRS-aided systems. The reason is that the power of the incident signals at each IRS element becomes smaller as the transmission distance increases, this rendering the maximum allowed amplification coefficient at each element of the active IRS to become larger. The larger amplification coefficient can effectively compensate the performance loss induced by increasing the transmission distance. Thus, the transmission coverage increases significantly by deploying the active IRS. This phenomenon demonstrates the effectiveness of the active IRS for achieving coverage extension compared to that of the conventional passive IRS.
\subsubsection{Impact of Deployment}
\begin{figure}
\begin{minipage}[t]{0.45\linewidth}
\centering
\includegraphics[width=2.9in]{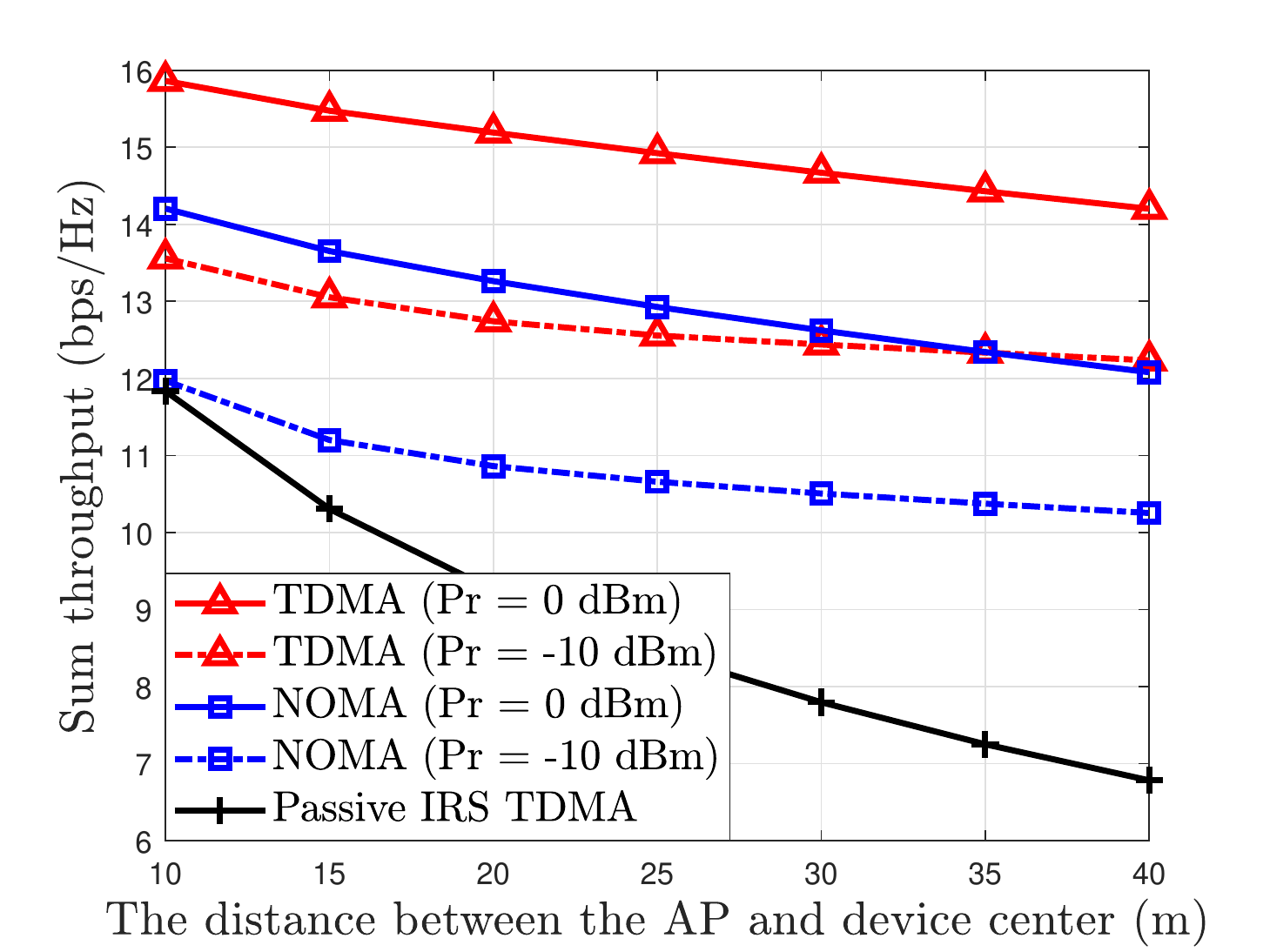}
\caption{Sum throughput versus the distance between AP and device center with ${E_k} = 0.01{\rm{ J}}$ and $K = 10$.}
\label{coverage}
\end{minipage}%
\hfill
\begin{minipage}[t]{0.45\linewidth}
\centering
\includegraphics[width=2.9in]{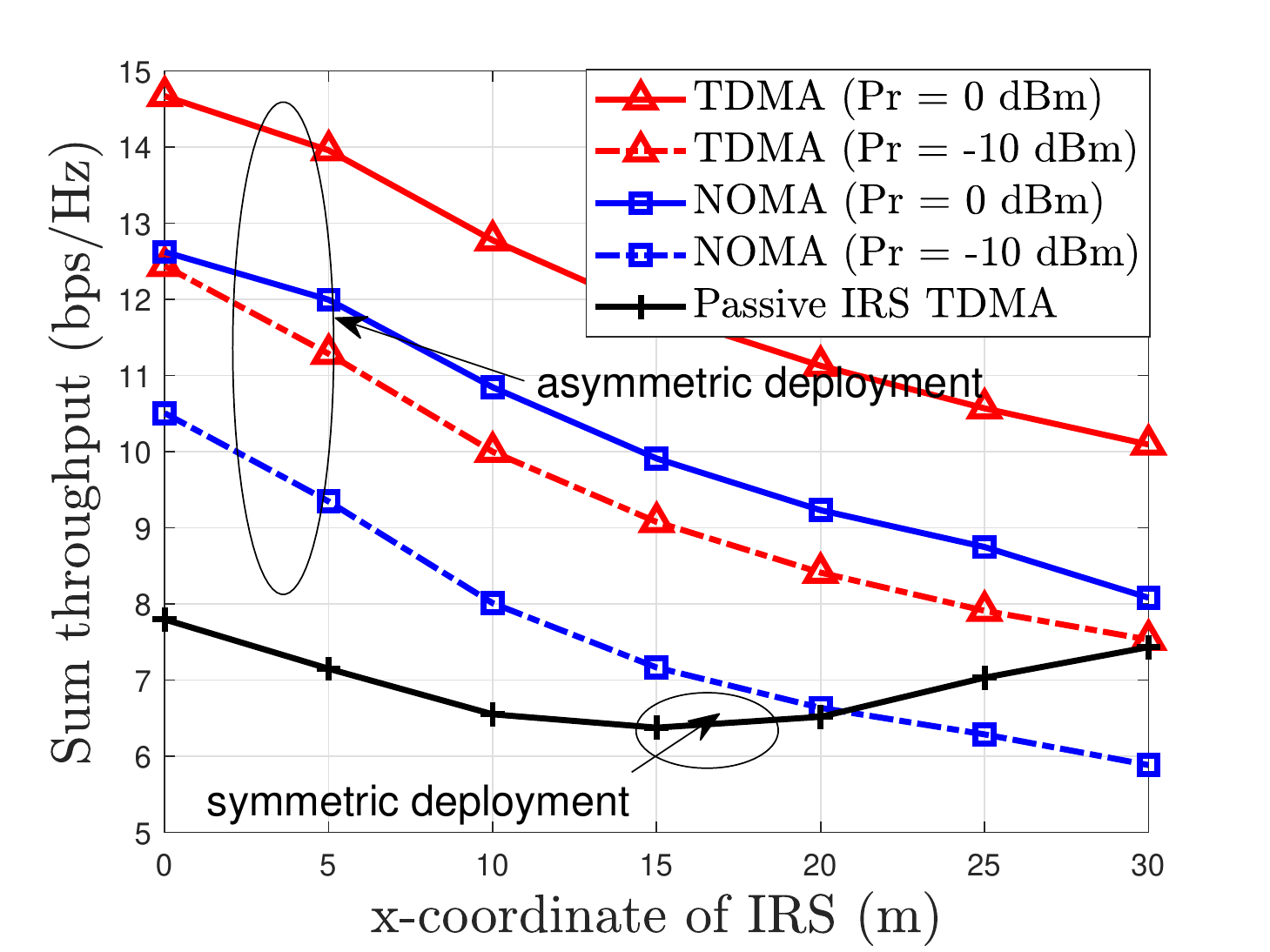}
\caption{Sum throughput versus the location of IRS with ${E_k} = 0.01{\rm{ J}}$ and $K = 10$.}
\label{deployment}
\end{minipage}
\vspace{-16pt}
\end{figure}
In Fig. \ref{deployment}, the sum throughput versus the x-coordinate of the IRS ${x_{{\rm{IRS}}}}$ is plotted to demonstrate the impact of the IRS deployment on the performance of our considered systems. For conventional passive IRS-aided systems, it was shown in \cite{9326394} that the IRS should be deployed in the close proximity to the AP or the device. While for the active IRS, it can be observed that the sum throughput decreases significantly when moving the IRS farther from the AP for both TDMA and NOMA cases. This is because the received signal power at the IRS gets stronger as the IRS comes closer to the devices (transmitters), which greatly limits the amplification gains at IRS elements. The main bottleneck restricting the performance becomes the deep fading of the IRS-AP link in this case. This phenomenon implies that for the active IRS, it is better to deploy it close to the AP (receiver). Additionally, the results indicate that different MA schemes have no impact on the deployment of the active IRS. Moreover, one can observe that the performance of active IRS-aided systems is even worse than that of the passive IRS if the position of the IRS is not properly selected, which highlights the importance of the active IRS deployment to unlock its full potential in maximizing the sum throughput.

\subsection{Performance Evaluation for Hybrid MA Schemes}
\begin{figure}
\begin{minipage}[t]{0.45\linewidth}
\centering
\includegraphics[width=2.9in]{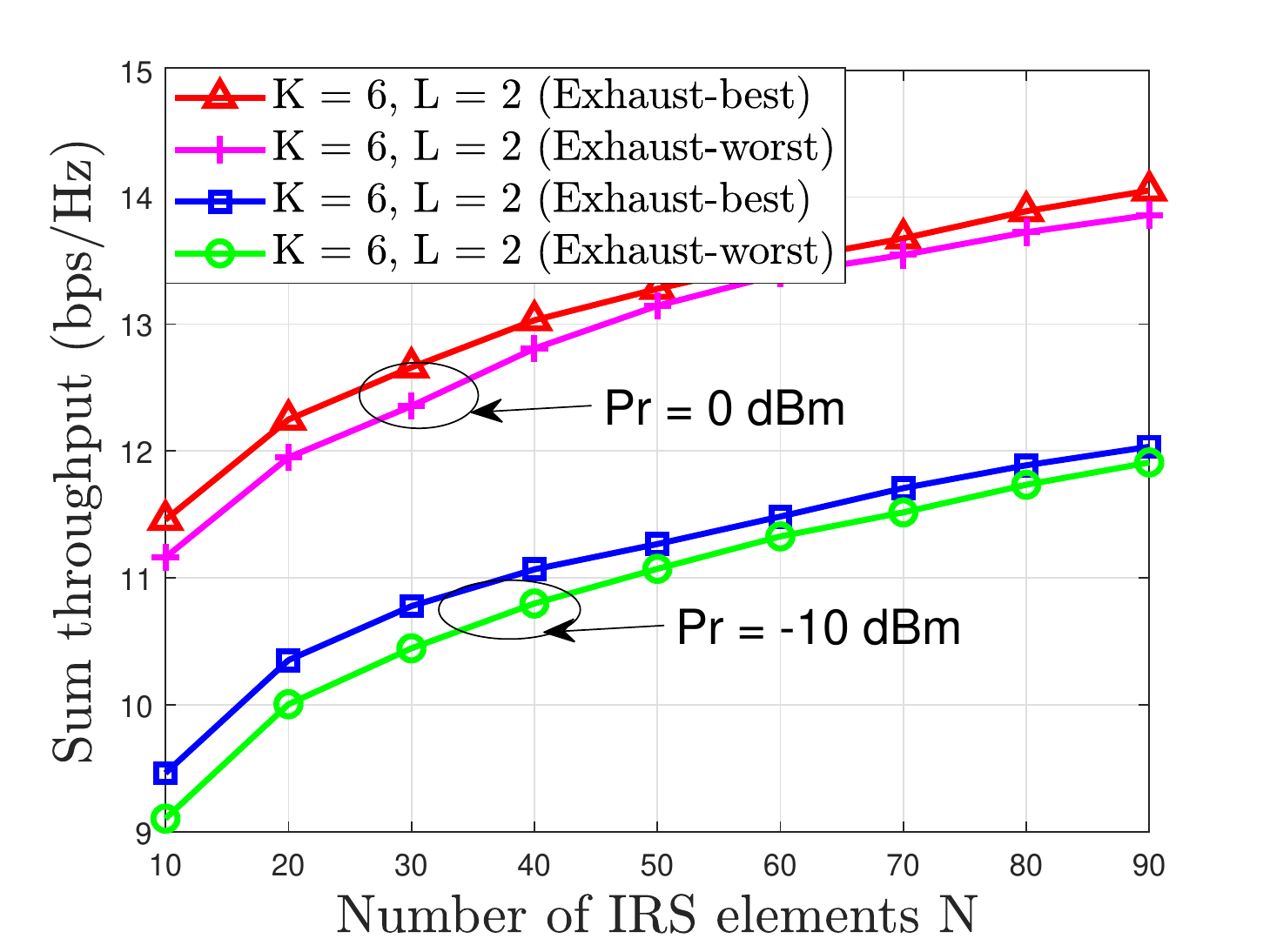}
\caption{Impact of user grouping on the sum throughput with ${E_k} = 0.01{\rm{ J}}$ and $K = 6$.}
\label{grouping}
\end{minipage}%
\hfill
\begin{minipage}[t]{0.45\linewidth}
\centering
\includegraphics[width=2.9in]{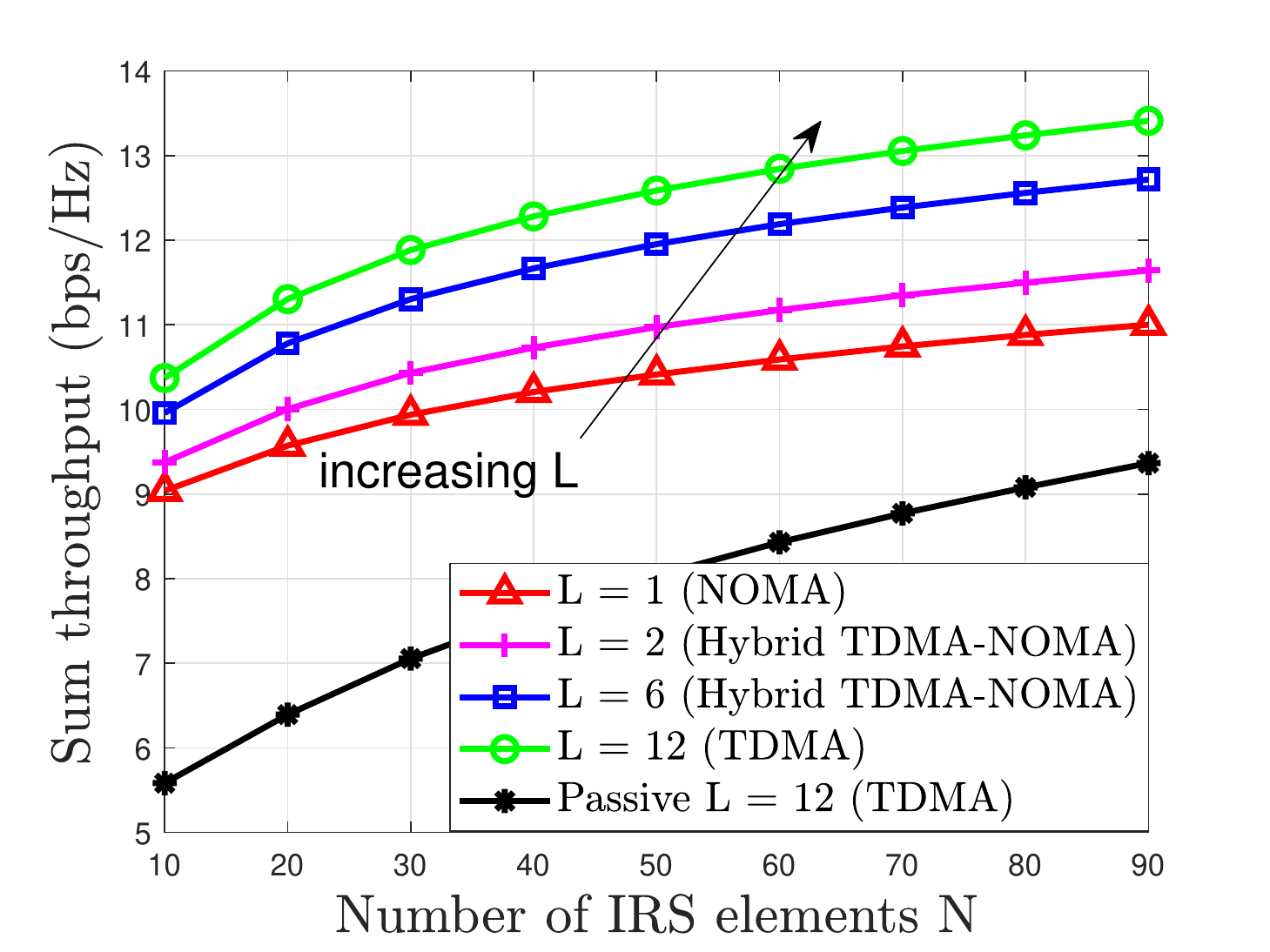}
\caption{Performance evaluation for hybrid MA schemes with ${E_k} = 0.01{\rm{ J}}$ and $K = 12$.}
\label{hybrid}
\end{minipage}
\vspace{-16pt}
\end{figure}
In this subsection, we provide numerical results for evaluating the performance of our proposed hybrid TDMA-NOMA scheme in terms of the sum throughput under different setups.
\subsubsection{Impact of User Grouping}
Before discussing the system performance of the proposed hybrid TDMA-NOMA scheme, we first shed light on the impact of user grouping. Specifically, we compare the following two cases: 1) ``Exhaust-best'' denotes the case where the optimal user grouping set is obtained through exhaustive search; 2) ``Exhaust-worst'' denotes the case where the IRS beamforming and resource allocation are optimized with the worst user grouping set obtained by exhaustive search. As shown in Fig. \ref{grouping}, ``Exhaust-best'' only achieves marginal gains (less than 2\%) over ``Exhaust-worst'' for both the cases of ${P_r} = 0$ dBm and ${P_r} = -10$ dBm, which implies that user grouping is not significant in our considered system. It is worth noting that developing sophisticated algorithms to obtain the optimal user grouping will be extremely time-consuming and computation-expensive, especially when $K$ becomes larger. Regarding an overloaded scenario, the results suggest that the random user grouping method can be applied in our considered systems with a negligible performance loss.

\subsubsection{Throughput Evaluation for Hybrid MA Schemes}
In Fig. \ref{hybrid}, we evaluate the performance of the hybrid TDMA-NOMA scheme, by plotting the system sum throughput versus the number of IRS elements under the different number of user groups, i.e., $L$. It is observed that the sum throughput achieved by our proposed algorithm increases as $L$ increases. The results are expected since exploiting more IRS beamforming vectors is indeed beneficial for the throughput improvement of our considered energy-constrained IoT systems. Furthermore, one can observe that the performance gap between $L = 6$ and $L = 12$ is slight (less than 5\%). Note that the number of IRS beamforming coefficients to be sent from the AP to the IRS controller is given by $LN$. By setting two devices in each group, the signalling overhead can be reduced by 50\% compared to the TDMA-based scheme at the cost of slight performance loss, which renders the hybrid TDMA-NOMA scheme a promising approach to balance the performance-overhead tradeoff. Finally, we observe that the active IRS aided system employing NOMA significantly outperforms that of the passive IRS-aided system employing TDMA while maintaining lower signalling overhead, which highlights the potential benefits of active IRS architectures for achieving both the higher throughput and lower signaling overhead compared to the conventional passive IRS.

\section{Conclusion}
Considering different MA schemes, we investigated an active IRS aided energy-constrained IoT system to maximize the sum throughput by jointly optimizing IRS beamforming vectors and resource allocation. Specifically, we first studied a couple of MA schemes, namely TDMA and NOMA. By deeply exploiting inherent properties of their associated optimization problems, we proposed two dedicated algorithms to solve them efficiently. Moreover, the theoretical performance comparison for the active IRS aided TDMA and NOMA schemes was provided. The results demonstrate that TDMA can potentially achieve a higher throughput than that of NOMA at the cost of more IRS beamforming vectors. Aiming at providing high flexibility in balancing the performance and signalling overhead tradeoff, we further develop a novel hybrid TDMA-NOMA scheme, which is applicable for any given number of IRS beamforming vectors available. The AO-based algorithm was extended to solve its associated sum throughput maximization problem. Numerical results validated our theoretical findings and unveiled the effectiveness of the active IRS architecture over the conventional passive IRS in terms of extending coverage range, reducing the requirement of reflecting elements, and supporting multiple low-energy IoT devices. Moreover, using the hybrid TDMA-NOMA scheme for assisting data transmission can be a practically appealing approach for flexibly balancing performance-overhead tradeoff, especially for IoT networks with practically large number of devices. The results in this paper demonstrated that integrating hybrid TDMA-NOMA scheme into active IRS-aided energy-constrained IoT systems is a promising solution for supporting massive number of energy-constrained IoT devices.
\section*{Appendix A: \textsc{Proof of Lemma 1}}
We first prove that all devices will be scheduled at the optimal solution, i.e., $\tau _k^* > 0,\forall k$. We show it by contradiction. Suppose that ${\Xi ^*} = \left\{ {\tau _k^*,p_k^*,{\bf{v}}_k^*} \right\}$ achieves the optimal solution of problem \eqref{C7} and there exists a device $j$ who will not be scheduled, i.e., $\tau _j^* = 0$. Then, we construct a different solution $\tilde \Xi  = \left\{ {{{\tilde \tau }_k},{{\tilde p}_k},{{{\bf{\tilde v}}}_k}} \right\}$, where ${{{\bf{\tilde v}}}_k} = {\bf{v}}_k^*\left( {k \ne j} \right)$, ${{{\bf{\tilde v}}}_j} = {\bf{0}}$, and
\begin{align}\label{new_construct_time}
{{\tilde \tau }_k} = \begin{cases}{\tau _k^*}, &{k \ne j,{\rm{ }}k \ne m}, k \in {\cal K}, \cr {\tau _k^* - \Delta {\tau _k}},
&{k = m}, m \in {\cal K}, \cr {\Delta {\tau _k}}, &{k = j}, j \in {\cal K}, \end{cases}
\end{align}
\begin{align}\label{new_construct_power}
{{\tilde p}_k} = \begin{cases}p_k^*, &{k \ne j}, k \in {\cal K}, \cr \frac{{{E_j}}}{{\Delta {\tau _m}}},
&{k = j},  j \in {\cal K}, m \in {\cal K}. \end{cases}
\end{align}
It can be readily verified that the newly constructed solution ${\tilde \Xi }$ is also a feasible solution for problem \eqref{C7} since it satisfies all the constraints therein. Since the solutions regarding the transmit power, the time allocation, and beamforming vector for any device ${k \ne j,{\rm{ }}k \ne m}$ remain unchanged in ${\Xi ^*}$ and ${\tilde \Xi }$, we only need to compare the throughput contributed by device $j$ and device $m$ for the corresponding two solutions. For the solution ${\Xi ^*}$, we have
\begin{align}\label{throughput_1}
R_m^* \!+\! R_j^* \!\!=\!\! \left( {\tau _m^* \!\!-\!\! \Delta {\tau _m}} \right){\log _2}\left( {1 \!\!+\!\! \frac{{{p_m}{{\left| {{h_{d,k}} + {\bf{v}}_m^H{{\bf{q}}_m}} \right|}^2}}}{{{\sigma ^2} \!\!+\!\! \sigma _r^2{\bf{v}}_m^H{\bf{G}}{{\bf{v}}_m}}}} \right) \!+\! \Delta {\tau _m}{\log _2}\left( {1 \!+\! \frac{{{p_m}{{\left| {{h_{d,k}} + {\bf{v}}_m^H{{\bf{q}}_m}} \right|}^2}}}{{{\sigma ^2} + \sigma _r^2{\bf{v}}_m^H{\bf{G}}{{\bf{v}}_m}}}} \right).
\end{align}
For the newly constructed solution ${\tilde \Xi }$, we have
\begin{align}\label{throughput_2}
{{\tilde R}_m} + {{\tilde R}_j} = \left( {\tau _m^* - \Delta {\tau _m}} \right){\log _2}\left( {1 + \frac{{{p_m}{{\left| {{h_{d,k}} + {\bf{v}}_m^H{{\bf{q}}_m}} \right|}^2}}}{{{\sigma ^2} + \sigma _r^2{\bf{v}}_m^H{\bf{G}}{{\bf{v}}_m}}}} \right) + \Delta {\tau _m}{\log _2}\left( {1 + \frac{{{E_j}}}{{\Delta {\tau _m}}}{{\left| {{h_{d,j}}} \right|}^2}} \right).
\end{align}
For a given positive value
\begin{align}\label{condition}
\varepsilon  = \frac{{{E_j}\left( {{\sigma ^2} + \sigma _r^2{\bf{v}}_m^H{\bf{G}}{{\bf{v}}_m}} \right){{\left| {{h_{d,j}}} \right|}^2}}}{{{p_m}{{\left| {{h_{d,k}} + {\bf{v}}_m^H{{\bf{q}}_m}} \right|}^2}}},
\end{align}
it can be easily verified that ${{\tilde R}_m} + {{\tilde R}_j} > R_m^* + R_j^*$ when $\Delta {\tau _m} < \varepsilon$. This means that the newly constructed solution ${\tilde \Xi }$ achieves a higher throughput than ${\Xi ^*}$, which contradicts the assumption that ${\Xi ^*}$ is the optimal solution. As such, we have $\tau _k^* > 0,\forall k$.

We next prove $p_k^* = {{{E_k}} \mathord{\left/
 {\vphantom {{{E_k}} {\tau _k^*}}} \right.
 \kern-\nulldelimiterspace} {\tau _k^*}}$ by contradiction. Suppose that $\left\{ {p_k^*,{\bf{v}}_k^*} \right\}$ is the optimal transmit power and IRS beamforming vector for device $k$ and $p_k^*{{ < {E_k}} \mathord{\left/
 {\vphantom {{ < {E_k}} {\tau _k^*}}} \right.
 \kern-\nulldelimiterspace} {\tau _k^*}}$. We can always construct a new solution denoted by $\left\{ {{{\tilde p}_k},{{{\bf{\tilde v}}}_k}} \right\}$ which satisfies $p_k^* < {{\tilde p}_k} \le {{{E_k}} \mathord{\left/
 {\vphantom {{{E_k}} {\tau _k^*}}} \right.
 \kern-\nulldelimiterspace} {\tau _k^*}}$ and ${{{\bf{\tilde v}}}_k} = \sqrt {{{p_k^*} \mathord{\left/
 {\vphantom {{p_k^*} {{{\tilde p}_k}}}} \right.
 \kern-\nulldelimiterspace} {{{\tilde p}_k}}}} {\bf{v}}_k^*$. It can be readily verified that
 \begin{align}\label{constraint_verify}
{{\tilde p}_k}{\bf{\tilde v}}_k^H{{\bf{H}}_{r,k}}{{{\bf{\tilde v}}}_k} + \sigma _r^2{\left\| {{{{\bf{\tilde v}}}_k}} \right\|^2} < p_k^*{\bf{v}}_k^{*H}{{\bf{H}}_{r,k}}{\bf{v}}_k^* + \sigma _r^2{\left\| {{\bf{v}}_k^*} \right\|^2} \le {P_r},
\end{align}
which indicates that $\left\{ {{{\tilde p}_k},{{{\bf{\tilde v}}}_k}} \right\}$ is feasible for problem \eqref{C7}. We further compare the objective values for the two solutions, namely $\left\{ {p_k^*,{\bf{v}}_k^*} \right\}$ and $\left\{ {{{\tilde p}_k},{{{\bf{\tilde v}}}_k}} \right\}$, as follows
 \begin{align}\label{obj_value_compare}
\frac{{p_k^*{{\left| {{h_{d,k}} + {\bf{v}}_k^{*H}{{\bf{q}}_k}} \right|}^2}}}{{{\sigma ^2} + \sigma _r^2{\bf{v}}_k^{*H}{\bf{Gv}}_k^*}}\mathop  < \limits^{\left( a \right)} \frac{{{{\tilde p}_k}{{\left| {{h_{d,k}} + {\bf{\tilde v}}_k^H{{\bf{q}}_k}} \right|}^2}}}{{{\sigma ^2} + \sigma _r^2{\bf{\tilde v}}_k^H{\bf{G}}{{{\bf{\tilde v}}}_k}}},
\end{align}
where inequality ${\left( a \right)}$ holds due to $p_k^*{\left| {{h_{d,k}}} \right|^2} < {{\tilde p}_k}{\left| {{h_{d,k}}} \right|^2}$, $p_k^*{\mathop{\rm Re}\nolimits} \left( {{h_{d,k}}{\bf{v}}_k^{*H}{{\bf{q}}_k}} \right) < {{\tilde p}_k}{\mathop{\rm Re}\nolimits} \left( {{h_{d,k}}{\bf{\tilde v}}_k^H{{\bf{q}}_k}} \right)$, $p_k^*{\left| {{\bf{v}}_k^{*H}{{\bf{q}}_k}} \right|^2} = {{\tilde p}_k}{\left| {{\bf{\tilde v}}_k^H{{\bf{q}}_k}} \right|^2}$, and ${\bf{v}}_k^{*H}{\bf{Gv}}_k^* > {\bf{\tilde v}}_k^H{\bf{G}}{{{\bf{\tilde v}}}_k}$. This means that the newly constructed solution $\left\{ {{{\tilde p}_k},{{{\bf{\tilde v}}}_k}} \right\}$ achieves a higher objective value than that of $\left\{ {p_k^*,{\bf{v}}_k^*} \right\}$, which contradicts the assumption that $\left\{ {p_k^*,{\bf{v}}_k^*} \right\}$ is optimal. Thus, the optimal transmit power satisfies $p_k^* = {{{E_k}} \mathord{\left/
 {\vphantom {{{E_k}} {\tau _k^*}}} \right.
 \kern-\nulldelimiterspace} {\tau _k^*}},\forall k$.
\section*{Appendix B: \textsc{Proof of Proposition 1}}
To prove Proposition 1, we first introduce a pair of problem formulations corresponding to problem \eqref{C11} and \eqref{C18} without constraints \eqref{C11-e} and \eqref{C18-c} as follows, respectively,
\begin{subequations}\label{C200}
\begin{align}
\label{C200-a}\mathop {\max }\limits_{{\tau},{p_k},{{\bf{v}}}}  \;\;&\tau {\log _2}\left( {1 + \frac{{\sum\nolimits_{k = 1}^K {{p_k}{{\left| {{h_{d,k}} + {{\bf{v}}^H}{{\bf{q}}_k}} \right|}^2}} }}{{{\sigma ^2} + \sigma _r^2{{\bf{v}}^H}{\bf{Gv}}}}} \right)\\
\label{C200-b}{\rm{s.t.}}\;\;&\eqref{C6-b}, \eqref{C6-c}, \eqref{C6-d}.
\end{align}
\end{subequations}
\begin{subequations}\label{C201}
\vspace{-11pt}
\begin{align}
\label{C201-a}\mathop {\max }\limits_{{\tau _k},{p_k},{{\bf{v}}}}  \;\;&\sum\limits_{k = 1}^K {{\tau _k}{{\log }_2}\left( {1 + \frac{{{p_k}{{\left| {{h_{d,k}} + {{\bf{v}}^H}{{\bf{q}}_k}} \right|}^2}}}{{{\sigma ^2} + \sigma _r^2{{\bf{v}}^H}{\bf{Gv}}}}} \right)}\\
\label{C201-b}{\rm{s.t.}}\;\;\;&\eqref{C5-b}, \eqref{C5-c}, \eqref{C5-d}.
\end{align}
\end{subequations}
The optimal values of problem \eqref{C200} and \eqref{C201} are denoted by $R_{{\rm{NOMA}}}^{{\rm{ub}}}$ and $R_{{\rm{TDMA}}}^{{\rm{ub}}}$, respectively. Firstly, we can prove $R_{{\rm{TDMA}}}^{{\rm{ub}}} = R_{{\rm{NOMA}}}^{{\rm{ub}}}$ by showing $R_{{\rm{TDMA}}}^{{\rm{ub}}} \le R_{{\rm{NOMA}}}^{{\rm{ub}}}$ and $R_{{\rm{TDMA}}}^{{\rm{ub}}} \ge R_{{\rm{NOMA}}}^{{\rm{ub}}}$.

The procedure starts by showing that $R_{{\rm{TDMA}}}^{{\rm{ub}}} \le R_{{\rm{NOMA}}}^{{\rm{ub}}}$. For problem \eqref{C201}, the optimal transmit power of each device, denoted by ${{\mathord{\buildrel{\lower3pt\hbox{$\scriptscriptstyle\smile$}}
\over p} }_k}$, can be expressed as ${{\mathord{\buildrel{\lower3pt\hbox{$\scriptscriptstyle\smile$}}
\over p} }_k} = {{{E_k}} \mathord{\left/
 {\vphantom {{{E_k}} {{\tau _k}}}} \right.
 \kern-\nulldelimiterspace} {{\tau _k}}}$, because each device will deplete all of its energy. To this end, we discuss some properties about the optimal time allocation for problem \eqref{C201}, i.e., ${{\mathord{\buildrel{\lower3pt\hbox{$\scriptscriptstyle\smile$}}
\over \tau } }_k}$. Given the optimal IRS beamforming vector ${\bf{v}} = {\bf{\mathord{\buildrel{\lower3pt\hbox{$\scriptscriptstyle\smile$}}
\over v} }}$, problem \eqref{C201} can be simplified by optimizing ${{\tau _k}}$ as follows
\begin{subequations}\label{C202}
\begin{align}
\label{C202-a}\mathop {\max }\limits_{{\tau _k}}  \;\;&\sum\limits_{k = 1}^K {{\tau _k}{{\log }_2}\left( {1 + \frac{{{E_k}{{\left| {{h_{d,k}} + {{{\bf{\mathord{\buildrel{\lower3pt\hbox{$\scriptscriptstyle\smile$}}
\over v} }}}^H}{{\bf{q}}_k}} \right|}^2}}}{{{\tau _k}\left( {{\sigma ^2} + \sigma _r^2{{{\bf{\mathord{\buildrel{\lower3pt\hbox{$\scriptscriptstyle\smile$}}
\over v} }}}^H}{\bf{G\mathord{\buildrel{\lower3pt\hbox{$\scriptscriptstyle\smile$}}
\over v} }}} \right)}}} \right)} \\
\label{C202-b}{\rm{s.t.}}\;\;&\eqref{C5-c}.
\end{align}
\end{subequations}
Note that problem \eqref{C202} is a convex optimization problem and its Lagrangian function is
\begin{align}\label{Lagrangian function}
{\cal L}\left( {{\tau _k},\lambda } \right) = \sum\limits_{k = 1}^K {{\tau _k}{{\log }_2}\left( {1 + \frac{{{E_k}{{\left| {{h_{d,k}} + {{{\bf{\mathord{\buildrel{\lower3pt\hbox{$\scriptscriptstyle\smile$}}
\over v} }}}^H}{{\bf{q}}_k}} \right|}^2}}}{{{\tau _k}\left( {{\sigma ^2} + \sigma _r^2{{{\bf{\mathord{\buildrel{\lower3pt\hbox{$\scriptscriptstyle\smile$}}
\over v} }}}^H}{\bf{G\mathord{\buildrel{\lower3pt\hbox{$\scriptscriptstyle\smile$}}
\over v} }}} \right)}}} \right)}  + \lambda \left( {{T_{\max }} - \sum\limits_{k = 1}^K {{\tau _k}} } \right),
\end{align}
where $\lambda  \ge 0$ is the dual variable associated with \eqref{C202-b}. According to Karush-Kuhn-Tucker (KKT) conditions, we have
\begin{align}\label{Lagrangian function}
\frac{{\partial {\cal L}\left( {{\tau _k},\lambda } \right)}}{{\partial {\tau _k}}} = \Gamma \left( {{\Upsilon _k}} \right) \buildrel \Delta \over = {\log _2}\left( {1 + {\Upsilon _k}} \right) - \frac{{{\Upsilon _k}}}{{\left( {1 + {\Upsilon _k}} \right)\ln 2}} - \lambda  = 0,
\end{align}
where
\begin{align}\label{gamma_K}
{\Upsilon _k} = \frac{{{E_k}{{\left| {{h_{d,k}} + {{{\bf{\mathord{\buildrel{\lower3pt\hbox{$\scriptscriptstyle\smile$}}
\over v} }}}^H}{{\bf{q}}_k}} \right|}^2}}}{{{\tau _k}\left( {{\sigma ^2} + \sigma _r^2{{{\bf{\mathord{\buildrel{\lower3pt\hbox{$\scriptscriptstyle\smile$}}
\over v} }}}^H}{\bf{G\mathord{\buildrel{\lower3pt\hbox{$\scriptscriptstyle\smile$}}
\over v} }}} \right)}},
\end{align}
can be regarded as the received signal-to-noise ratio (SNR) of device $k$. Since $\Gamma \left( {{\Upsilon _k}} \right)$ is an increasing function with respect to ${{\Upsilon _k}}$ and $\Gamma \left( 0 \right) \le 0$, equation $\Gamma \left( {{\Upsilon _k}} \right) = 0$ has a unique solution, which implies that all devices share the same SNR at the optimal solution, i.e.,
\begin{align}\label{SNR_relation}
\frac{{{E_k}{{\left| {{h_{d,k}} + {{{\bf{\mathord{\buildrel{\lower3pt\hbox{$\scriptscriptstyle\smile$}}
\over v} }}}^H}{{\bf{q}}_k}} \right|}^2}}}{{{{\mathord{\buildrel{\lower3pt\hbox{$\scriptscriptstyle\smile$}}
\over \tau } }_k}\left( {{\sigma ^2} + \sigma _r^2{{{\bf{\mathord{\buildrel{\lower3pt\hbox{$\scriptscriptstyle\smile$}}
\over v} }}}^H}{\bf{G\mathord{\buildrel{\lower3pt\hbox{$\scriptscriptstyle\smile$}}
\over v} }}} \right)}} = \frac{{{E_j}{{\left| {{h_{d,k}} + {{{\bf{\mathord{\buildrel{\lower3pt\hbox{$\scriptscriptstyle\smile$}}
\over v} }}}^H}{{\bf{q}}_j}} \right|}^2}}}{{{{\mathord{\buildrel{\lower3pt\hbox{$\scriptscriptstyle\smile$}}
\over \tau } }_j}\left( {{\sigma ^2} + \sigma _r^2{{{\bf{\mathord{\buildrel{\lower3pt\hbox{$\scriptscriptstyle\smile$}}
\over v} }}}^H}{\bf{G\mathord{\buildrel{\lower3pt\hbox{$\scriptscriptstyle\smile$}}
\over v} }}} \right)}},\forall k,j \in {\cal K}.
\end{align}
As such, the optimal value of problem \eqref{C201} can be written as
\begin{align}\label{obj_value_TDMA}
R_{{\rm{TDMA}}}^{{\rm{ub}}} = \mathord{\buildrel{\lower3pt\hbox{$\scriptscriptstyle\smile$}}
\over \tau } {\log _2}\left( {1 + \frac{{\sum\nolimits_{k = 1}^K {{E_k}{{\left| {{h_{d,k}} + {{{\bf{\mathord{\buildrel{\lower3pt\hbox{$\scriptscriptstyle\smile$}}
\over v} }}}^H}{{\bf{q}}_k}} \right|}^2}} }}{{\mathord{\buildrel{\lower3pt\hbox{$\scriptscriptstyle\smile$}}
\over \tau } \left( {\left( {{\sigma ^2} + \sigma _r^2{{{\bf{\mathord{\buildrel{\lower3pt\hbox{$\scriptscriptstyle\smile$}}
\over v} }}}^H}{\bf{G\mathord{\buildrel{\lower3pt\hbox{$\scriptscriptstyle\smile$}}
\over v} }}} \right)} \right)}}} \right),
\end{align}
where $\mathord{\buildrel{\lower3pt\hbox{$\scriptscriptstyle\smile$}}
\over \tau }  = \sum\nolimits_{k = 1}^K {{{\mathord{\buildrel{\lower3pt\hbox{$\scriptscriptstyle\smile$}}
\over \tau } }_k}}$. It can be verified that $\left\{ {\tau  = \mathord{\buildrel{\lower3pt\hbox{$\scriptscriptstyle\smile$}}
\over \tau } ,{p_k} = {{{E_k}} \mathord{\left/
 {\vphantom {{{E_k}} {\mathord{\buildrel{\lower3pt\hbox{$\scriptscriptstyle\smile$}}
\over \tau } }}} \right.
 \kern-\nulldelimiterspace} {\mathord{\buildrel{\lower3pt\hbox{$\scriptscriptstyle\smile$}}
\over \tau } }},{\bf{v}} = {\bf{\mathord{\buildrel{\lower3pt\hbox{$\scriptscriptstyle\smile$}}
\over v} }}} \right\}$ is one feasible solution for problem \eqref{C200}, which yields $R_{{\rm{TDMA}}}^{{\rm{ub}}} \le R_{{\rm{NOMA}}}^{{\rm{ub}}}$.

Next, we show that $R_{{\rm{TDMA}}}^{{\rm{ub}}} \ge R_{{\rm{NOMA}}}^{{\rm{ub}}}$. For problem \eqref{C200}, it can be easily shown that each device will deplete all of its energy and thus the optimal solution of problem \eqref{C200} is denoted by $\left\{ {\tau  = \mathord{\buildrel{\lower3pt\hbox{$\scriptscriptstyle\frown$}}
\over \tau } ,{p_k} = {{{E_k}} \mathord{\left/
 {\vphantom {{{E_k}} {\mathord{\buildrel{\lower3pt\hbox{$\scriptscriptstyle\frown$}}
\over \tau } }}} \right.
 \kern-\nulldelimiterspace} {\mathord{\buildrel{\lower3pt\hbox{$\scriptscriptstyle\frown$}}
\over \tau } }},{\bf{v}} = {\bf{\mathord{\buildrel{\lower3pt\hbox{$\scriptscriptstyle\frown$}}
\over v} }}} \right\}$. Based on the optimal solution of problem \eqref{C200}, we can always construct a new solution, which satisfies $\mathord{\buildrel{\lower3pt\hbox{$\scriptscriptstyle\frown$}}
\over \tau }  = \sum\nolimits_{k = 1}^K {{{\mathord{\buildrel{\lower3pt\hbox{$\scriptscriptstyle\frown$}}
\over \tau } }_k}}$ and ${{\mathord{\buildrel{\lower3pt\hbox{$\scriptscriptstyle\frown$}}
\over p} }_k} = {{{E_k}} \mathord{\left/
 {\vphantom {{{E_k}} {{{\mathord{\buildrel{\lower3pt\hbox{$\scriptscriptstyle\frown$}}
\over \tau } }_k}}}} \right.
 \kern-\nulldelimiterspace} {{{\mathord{\buildrel{\lower3pt\hbox{$\scriptscriptstyle\frown$}}
\over \tau } }_k}}}$, so that all devices share the same received SNR in the TDMA scheme, i.e.,
\begin{align}\label{NOMA_SNR}
\frac{{{E_k}{{\left| {{h_{d,k}} + {{{\bf{\mathord{\buildrel{\lower3pt\hbox{$\scriptscriptstyle\frown$}}
\over v} }}}^H}{{\bf{q}}_k}} \right|}^2}}}{{{{\mathord{\buildrel{\lower3pt\hbox{$\scriptscriptstyle\frown$}}
\over \tau } }_k}\left( {{\sigma ^2} + \sigma _r^2{{{\bf{\mathord{\buildrel{\lower3pt\hbox{$\scriptscriptstyle\frown$}}
\over v} }}}^H}{\bf{G\mathord{\buildrel{\lower3pt\hbox{$\scriptscriptstyle\frown$}}
\over v} }}} \right)}} = \frac{{{E_j}{{\left| {{h_{d,k}} + {{{\bf{\mathord{\buildrel{\lower3pt\hbox{$\scriptscriptstyle\frown$}}
\over v} }}}^H}{{\bf{q}}_j}} \right|}^2}}}{{{{\mathord{\buildrel{\lower3pt\hbox{$\scriptscriptstyle\frown$}}
\over \tau } }_j}\left( {{\sigma ^2} + \sigma _r^2{{{\bf{\mathord{\buildrel{\lower3pt\hbox{$\scriptscriptstyle\frown$}}
\over v} }}}^H}{\bf{G\mathord{\buildrel{\lower3pt\hbox{$\scriptscriptstyle\frown$}}
\over v} }}} \right)}},\forall k,j \in {\cal K}.
\end{align}
Note that $\left\{ {{\tau _k} = {{\mathord{\buildrel{\lower3pt\hbox{$\scriptscriptstyle\frown$}}
\over \tau } }_k},{p_k} = {{{E_k}} \mathord{\left/
 {\vphantom {{{E_k}} {{{\mathord{\buildrel{\lower3pt\hbox{$\scriptscriptstyle\frown$}}
\over \tau } }_k}}}} \right.
 \kern-\nulldelimiterspace} {{{\mathord{\buildrel{\lower3pt\hbox{$\scriptscriptstyle\frown$}}
\over \tau } }_k}}},{\bf{v}} = {\bf{\mathord{\buildrel{\lower3pt\hbox{$\scriptscriptstyle\frown$}}
\over v} }}} \right\}$ is also feasible for problem \eqref{C201} and always exists, which yields
\begin{align}\label{NOMA_throughput}
R_{{\rm{NOMA}}}^{{\rm{ub}}} = \sum\limits_{k = 1}^K {{{\mathord{\buildrel{\lower3pt\hbox{$\scriptscriptstyle\frown$}}
\over \tau } }_k}} {\log _2}\left( {1 + \frac{{{E_k}{{\left| {{h_{d,k}} + {{{\bf{\mathord{\buildrel{\lower3pt\hbox{$\scriptscriptstyle\frown$}}
\over v} }}}^H}{{\bf{q}}_k}} \right|}^2}}}{{{{\mathord{\buildrel{\lower3pt\hbox{$\scriptscriptstyle\frown$}}
\over \tau } }_k}\left( {{\sigma ^2} + \sigma _r^2{{{\bf{\mathord{\buildrel{\lower3pt\hbox{$\scriptscriptstyle\frown$}}
\over v} }}}^H}{\bf{G\mathord{\buildrel{\lower3pt\hbox{$\scriptscriptstyle\frown$}}
\over v} }}} \right)}}} \right).
\end{align}
Thus, at the optimal solution of problem \eqref{C201}, it follows that $R_{{\rm{TDMA}}}^{{\rm{ub}}} \ge R_{{\rm{NOMA}}}^{{\rm{ub}}}$.

Given $R_{{\rm{TDMA}}}^{{\rm{ub}}} \le R_{{\rm{NOMA}}}^{{\rm{ub}}}$ and $R_{{\rm{TDMA}}}^{{\rm{ub}}} \ge R_{{\rm{NOMA}}}^{{\rm{ub}}}$, we have $R_{{\rm{TDMA}}}^{{\rm{ub}}} = R_{{\rm{NOMA}}}^{{\rm{ub}}}$. Regarding their optimal solutions, we have $\mathord{\buildrel{\lower3pt\hbox{$\scriptscriptstyle\frown$}}
\over \tau }  = \sum\nolimits_{k = 1}^K {{{\mathord{\buildrel{\lower3pt\hbox{$\scriptscriptstyle\smile$}}
\over \tau } }_k}}$, ${\bf{\mathord{\buildrel{\lower3pt\hbox{$\scriptscriptstyle\frown$}}
\over v} }} = {\bf{\mathord{\buildrel{\lower3pt\hbox{$\scriptscriptstyle\smile$}}
\over v} }}$, and ${{\mathord{\buildrel{\lower3pt\hbox{$\scriptscriptstyle\frown$}}
\over e} }_k} = {{\mathord{\buildrel{\lower3pt\hbox{$\scriptscriptstyle\smile$}}
\over e} }_k}$, where ${{\mathord{\buildrel{\lower3pt\hbox{$\scriptscriptstyle\frown$}}
\over e} }_k} = \mathord{\buildrel{\lower3pt\hbox{$\scriptscriptstyle\frown$}}
\over \tau } {{\mathord{\buildrel{\lower3pt\hbox{$\scriptscriptstyle\frown$}}
\over p} }_k}$ and ${{\mathord{\buildrel{\lower3pt\hbox{$\scriptscriptstyle\smile$}}
\over e} }_k} = {{\mathord{\buildrel{\lower3pt\hbox{$\scriptscriptstyle\smile$}}
\over \tau } }_k}{{\mathord{\buildrel{\lower3pt\hbox{$\scriptscriptstyle\smile$}}
\over p} }_k}$. Then, we add constraints \eqref{C11-e} and \eqref{C18-c} in problems \eqref{C200} and \eqref{C201}, respectively, which correspondingly results in problems \eqref{C11} and \eqref{C18}. By letting ${e_k} = \tau {p_k}$ in problem \eqref{C11} and ${e_k} = {\tau _k}{p_k}$ in problem \eqref{C18}, constraints \eqref{C11-e} and \eqref{C18-c} can be equivalently rewritten as
\begin{subequations}
\begin{align}
\label{new_constraints1} &\sum\limits_{k = 1}^K {{e_k}{{\bf{v}}^H}{{\bf{H}}_{r,k}}{\bf{v}}}  + \tau \sigma _r^2{\left\| {\bf{v}} \right\|^2} \le {P_r}\tau,  \\
\label{new_constraints2} &{e_k}{{\bf{v}}^H}{{\bf{H}}_{r,k}}{\bf{v}} + {\tau _k}\sigma _r^2{\left\| {\bf{v}} \right\|^2} \le {P_r}{\tau _k}, \forall k \in {\cal K}.
\end{align}
\end{subequations}
For any feasible solution $\left\{ {{{\tilde \tau }_k},{{\tilde e}_k},{\bf{\tilde v}}} \right\}$ for problem \eqref{C18}, we have
\begin{align}\label{inequity}
&\sum\limits_{k = 1}^K {{{\tilde \tau }_k}{{\log }_2}\left( {1 + \frac{{{{\tilde e}_k}{{\left| {{h_{d,k}} + {{\bf{v}}^H}{{\bf{q}}_k}} \right|}^2}}}{{{{\tilde \tau }_k}\left( {{\sigma ^2} + \sigma _r^2{{\bf{v}}^H}{\bf{Gv}}} \right)}}} \right)}\mathop  \le \limits^{\left( a \right)} \sum\limits_{k = 1}^K {{{\tilde \tau }_k}{{\log }_2}\left( {1 + \frac{{\sum\nolimits_{k = 1}^K {{{\tilde e}_k}{{\left| {{h_{d,k}} + {{\bf{v}}^H}{{\bf{q}}_k}} \right|}^2}} }}{{\sum\nolimits_{k = 1}^K {{{\tilde \tau }_k}\left( {{\sigma ^2} + \sigma _r^2{{\bf{v}}^H}{\bf{Gv}}} \right)} }}} \right)}\nonumber\\
& = \tilde \tau {\log _2}\left( {1 + \frac{{\sum\nolimits_{k = 1}^K {{{\tilde e}_k}{{\left| {{h_{d,k}} + {{\bf{v}}^H}{{\bf{q}}_k}} \right|}^2}} }}{{\tilde \tau \left( {{\sigma ^2} + \sigma _r^2{{\bf{v}}^H}{\bf{Gv}}} \right)}}} \right),
\end{align}
where (a) follows from the inequality of arithmetic and geometric means, and $\tilde \tau  = \sum\nolimits_{k = 1}^K {{{\tilde \tau }_k}}$. We have $\sum\nolimits_{k = 1}^K {{{\tilde e}_k}{{{\bf{\tilde v}}}^H}{{\bf{H}}_{r,k}}{\bf{\tilde v}}}  + \tilde \tau \sigma _r^2{\left\| {{\bf{\tilde v}}} \right\|^2} \le {P_r}\tilde \tau $ by taking the summation of all $\forall k \in {\cal K}$ in \eqref{new_constraints2}, which implies that $\left\{ {\tilde \tau ,{{\tilde e}_k},{\bf{\tilde v}}} \right\}$ is guaranteed as a feasible solution for \eqref{C11}. As such, the optimal value of problem \eqref{C11} is no smaller than that of \eqref{C18}, i.e., $R_{{\rm{NOMA}}}^* \ge R_{{\rm{TDMA}}}^{\left( {{\rm{lb}}} \right)*}$. If
\begin{align}\label{inequity_condition}
\frac{{{{\tilde e}_k}{{\left| {{h_{d,k}} + {{\bf{v}}^H}{{\bf{q}}_k}} \right|}^2}}}{{{{\tilde \tau }_k}\left( {{\sigma ^2} + \sigma _r^2{{\bf{v}}^H}{\bf{Gv}}} \right)}} \ne \frac{{{{\tilde e}_j}{{\left| {{h_{d,j}} + {{\bf{v}}^H}{{\bf{q}}_j}} \right|}^2}}}{{{{\tilde \tau }_j}\left( {{\sigma ^2} + \sigma _r^2{{\bf{v}}^H}{\bf{Gv}}} \right)}},\exists k,j \in {\cal K},
\end{align}
we have $R_{{\rm{NOMA}}}^* > R_{{\rm{TDMA}}}^{\left( {{\rm{lb}}} \right)*}$, i.e., the optimal value of \eqref{C11} is strictly larger than that of \eqref{C18}.

When $\left\{ {{{\mathord{\buildrel{\lower3pt\hbox{$\scriptscriptstyle\smile$}}
\over \tau } }_k},{{\mathord{\buildrel{\lower3pt\hbox{$\scriptscriptstyle\smile$}}
\over e} }_k},{\bf{\mathord{\buildrel{\lower3pt\hbox{$\scriptscriptstyle\smile$}}
\over v} }}} \right\}$ satisfies constraint \eqref{new_constraints2}, it is also an optimal solution for problem \eqref{C11} and thus $R_{{\rm{TDMA}}}^{\left( {{\rm{lb}}} \right)*} = R_{{\rm{TDMA}}}^{{\rm{ub}}}$. In this case, $\mathord{\buildrel{\lower3pt\hbox{$\scriptscriptstyle\frown$}}
\over \tau }  = \sum\nolimits_{k = 1}^K {{{\mathord{\buildrel{\lower3pt\hbox{$\scriptscriptstyle\smile$}}
\over \tau } }_k}}$, ${\bf{\mathord{\buildrel{\lower3pt\hbox{$\scriptscriptstyle\frown$}}
\over v} }} = {\bf{\mathord{\buildrel{\lower3pt\hbox{$\scriptscriptstyle\smile$}}
\over v} }}$, and ${{\mathord{\buildrel{\lower3pt\hbox{$\scriptscriptstyle\frown$}}
\over e} }_k} = {{\mathord{\buildrel{\lower3pt\hbox{$\scriptscriptstyle\smile$}}
\over e} }_k}$ also satisfy all constraints in problem \eqref{C11}, which indicates that $R_{{\rm{TDMA}}}^{\left( {{\rm{lb}}} \right)*} = R_{{\rm{TDMA}}}^{{\rm{ub}}} = R_{{\rm{NOMA}}}^{{\rm{ub}}} = R_{{\rm{NOMA}}}^*$. Therefore, \eqref{sufficient_condition} is a sufficient condition for $R_{{\rm{TDMA}}}^{\left( {{\rm{lb}}} \right)*} = R_{{\rm{NOMA}}}^*$, which thus completes the proof.
\bibliographystyle{IEEEtran}
\bibliography{IEEEabrv,myref}


\end{document}